\documentclass[preprint,authoryear,12pt]{elsarticle}

\usepackage{floatrow}
\usepackage{amsmath}
\usepackage{amssymb}

\journal{arXiv}

\usepackage{amsmath}
\usepackage{amsthm}
\newtheorem{definition}{Definition}
\newtheorem{corollary}{Corollary}
\newtheorem{theorem}{Theorem}
\newtheorem*{theorem*}{Theorem}
\newtheorem{lemma}{Lemma}
\newtheorem{remark}{Remark}
\newtheorem{example}{Example}
\usepackage{optidef}
\usepackage{subcaption}

\allowdisplaybreaks

\usepackage{todonotes}

\usepackage{calc}
\usepackage{algorithm}
\usepackage{algpseudocode}
\usepackage{enumitem}
\usepackage{url}

\usepackage{hyperref}
\usepackage{breakurl}



\usepackage{mathtools}

\mathtoolsset{showonlyrefs}

\begin{document}

\begin{frontmatter}



\title{Randomized $p$-values for multiple testing and their application in replicability analysis}

\author{Anh-Tuan Hoang}
\author{Thorsten Dickhaus\corref{cor1}}
\cortext[cor1]{Institute for Statistics, University of Bremen, 
P. O. Box 330 440, 28344 Bremen, Germany. Tel: +49 421 218-63651. E-mail address: dickhaus@uni-bremen.de (Thorsten Dickhaus).}

\address{Institute for Statistics, University of Bremen, Germany}

\begin{abstract}
We are concerned with testing replicability hypotheses for many endpoints simultaneously. This constitutes a multiple test problem with composite null hypotheses. Traditional $p$-values, which are computed under least favourable parameter configurations, are over-conservative in the case of composite null hypotheses. As demonstrated in prior work, this poses severe challenges in the multiple testing context, especially when one goal of the statistical analysis is to estimate the proportion $\pi_0$ of true null hypotheses.
Randomized $p$-values have been proposed to remedy this issue. In the present work, we discuss the application of randomized $p$-values in replicability analysis. In particular, we introduce a general class of statistical models for which valid, randomized $p$-values can be calculated easily. By means of computer simulations, we demonstrate that their usage typically leads to a much more accurate estimation of $\pi_0$. Finally, we apply our proposed methodology to a real data example from genomics.
\end{abstract}

\end{frontmatter}


\section{Introduction} \label{sec:introduction}
The replication of scientific results is essential for their acceptance by the scientific community. In order to judge whether a scientific result has been replicated in an independent study, appropriate scientific methods are needed. We are concerned with developing such methods by formalizing the replication as a statistical hypothesis which has to be tested with an appropriate procedure. In particular, a simultaneous replicability analysis for many endpoints or markers, respectively, requires specialized multiple test procedures. We propose the usage of randomized $p$-values, as introduced by \cite{dickhaus2013randomized}, in this context.

For a single hypothesis test based on a test statistic $T(X)$, where $X$ is the observable random variable, mathematically representing the data set, a (non-randomized) $p$-value $p(X)$ is a (deterministic) transformation of $T(X)$ onto $[0,1]$. Small values of $p(X)$ indicate incompatibility of the observed data with the null hypothesis $H$ of interest. When basing test decisions on the $p$-value, type I error control at any pre-defined significance level $\alpha\in(0,1)$ is then equivalent to
\begin{equation} \label{eq:validityintro}
\text{pr}_{\vartheta}(p(X)\leq\alpha)\leq\alpha,\;\alpha\in(0,1),\vartheta\in H,
\end{equation}
where $\text{pr}_{\vartheta}$ denotes the probability measure under the parameter $\vartheta$ of the statistical model under consideration. A $p$-value fulfilling \eqref{eq:validityintro} is called a valid $p$-value.

In case of composite null hypotheses $H$, $p$-values are required to fulfill condition \eqref{eq:validityintro} under all parameter values $\vartheta_{0}\in H$. Hence, it is of interest to determine parameter values in $H$ which maximize the probability in \eqref{eq:validityintro}. These are called least favourable parameter configurations (LFCs). Under continuity assumptions, the $p$-value will usually be uniformly distributed under LFCs. However, if $\vartheta\in H$ is not an LFC, we typically have a strict inequality in condition \eqref{eq:validityintro} for many values of $\alpha \in (0,1)$. 

In the context of simultaneous testing of multiple null hypotheses, this deviation from the uniform distribution is problematic when utilizing data-adaptive multiple tests that rely on a pre-estimation of the proportion $\pi_{0}$ of the true null hypotheses. Non-uniformity can for example be caused by the presence of composite null hypotheses, as described before, or by the discreteness of the model. Randomized $p$-values resulting from a data-dependent mixing of the original $p$-value and an additional on $[0,1]$ uniformly distributed random variable $U$, that is stochastically independent of the data $X$, are then often considered in the literature. The distribution of the randomized $p$-values under the null is typically much closer to uniformity than that of the non-randomized ones. In case of discrete models randomized $p$-values for simple null hypotheses $H = \{\vartheta^*\}$ have been discussed, among others, by \cite{fistra2007,habiger2011randomised,RandoFisher,habiger2015multiple}. These randomized $p$-values are closely related to well-known randomized hypothesis tests in discrete models, and they are exactly uniformly distributed under $\vartheta^*$. For composite null hypotheses $H$, even in non-discrete models, it is generally not possible to achieve exact uniformity without abandoning the data completely. \cite{dickhaus2013randomized} proposed one set of data-dependent weights for the mixing of $X$ and $U$, that works well for composite, one-sided null hypotheses at least in certain location parameter models.

Due to irreproducibility, randomized $p$-values are not suitable for the final decision making. However, as demonstrated by \cite{RandoFisher}, \cite{dickhaus2013randomized} and others, they are very useful in the context of estimating the proportion $\pi_{0}$ of true null hypotheses. One popular estimator for $\pi_0$ has been proposed by \cite{schweder1982plots}. We will denote this estimator by $\hat{\pi}_{0} \equiv \hat{\pi}_{0}(\lambda)$, where $\lambda\in[0,1)$ is a tuning parameter, and will refer to $\hat{\pi}_{0}$ as the Schweder-Spjøtvoll estimator. The proposal is to utilize randomized $p$-values in $\hat{\pi}_{0}$. Since validity of the $p$-values utilized in $\hat{\pi}_{0}$ is essential for conservative estimations $\hat{\pi}_{0}$ (see Lemma $1$ in \cite{RandoFisher}), we will provide some sufficient conditions for the validity of randomized $p$-values in the sequel.

We will be particularly interested in replicability analysis, where one aims at identifying discoveries made across more than one of $s\geq 2$ given independent studies. The null hypothesis of no replication is a special type of a composite null hypothesis. While a typical meta-analysis (see, e.\ g., \cite{Kulinskaya}) pools the available data across the studies, replicability analysis requires findings to hold in at least $\gamma$ studies, where $2\leq\gamma\leq s$ is a pre-defined parameter. This is an important distinction, since in a meta-analysis, one extremely small $p$-value may suffice to produce a small combined $p$-value, regardless of the evidence contributed by the other studies. Instead of combining all (endpoint-specific) $p$-values from the $s$ studies, replicability analysis will usually apply a combination of all but the $\gamma-1$ smallest of these $p$-values. In the context of bio-marker identification we consider $s\geq 2$ independent studies that examine $m\geq 2$ endpoints as possible bio-markers for a given disease. Whether one endpoint constitutes a bio-marker may differ between the studies, since the latter are (usually) conducted under different settings like different (sub-)populations. It is of interest to find bio-markers that are associated with the disease in at least $\gamma$ different settings to rule out findings that can only be ascribed to one specific study setup. With our proposed methodology, it is possible to accurately estimate the number of replicated bio-markers. This is of interest in itself, but can also be used to increase the power of a multiple test for replicability. More details are provided in Sections \ref{sec:replicabilityanalysis} and \ref{sec:estimationtruenullhypotheses}.

Simultaneous testing of multiple replicability statements has also been the focus in prior literature. \cite{benjamini2009selective} made use of partial conjunction nulls, meaning that at least a pre-specified number of the (study-specific) null hypotheses for a given endpoint are true, see also \cite{benjamini2008screening}. They propose combining the $s-\gamma+1$ largest $p$-values for each endpoint in an appropriate manner, and then using an FDR controlling procedure on these partial conjunction $p$-values. \cite{bogomolov2013discovering} presented algorithms that separate $s=2$ studies into primary and follow-up study. An empirical Bayesian approach has been proposed by \cite{heller2014replicability}. \cite{heller2014deciding} introduced the $r$-value for each hypothesis, which indicates the lowest significance level with respect to the false discovery rate \citep{benjamini1995controlling} at which the corresponding hypothesis can be rejected. This  allows for a ranking among the examined features. \cite{bogomolov2018assessing} proposed to first select the promising features from each study separately and then to test the selected features.

\section{Model setup}
\label{sec:modelsetup}
In the following, we introduce a general model for which randomized $p$-values are easily computable. The parameter $s$ will be the number of studies, and the parameter $m\geq 2$ the number of endpoints (potential bio-markers) which also equals the number of null hypotheses. In the examples in Sections$~\ref{sec:modelsetup}$ and $\ref{sec:randomizedpvalues}$ we only consider the case of $s=1$, which can be interpreted as bio-marker identification without replicability requirements. In Section$~\ref{sec:replicabilityanalysis}$, where we introduce replicability analysis, we only consider $s\geq 2$.

Consider a statistical model $(\Omega,\mathcal{F},(\text{pr}_{\vartheta})_{\vartheta\in\Theta})$ and let $\theta=(\theta_{1},\ldots,\theta_{m}):\Theta\to\Theta '$ denote a derived parameter, in which $\Theta'=\Theta_{1}'\times\cdots\times\Theta_{m}'$ is a subset of $\mathbb{R}^{sm}=\mathbb{R}^{s}\times\cdots\times \mathbb{R}^{s}$, $s\geq 1$, where $\mathbb{R}$ denotes the set of real numbers. We assume that consistent and, at least asymptotically, unbiased estimators $\hat{\theta}_{j}=(\hat{\theta}_{1,j},\ldots,\hat{\theta}_{s,j}):\Omega\to \mathbb{R}^{s}$, for $\theta_{j}(\vartheta)=(\theta_{1,j}(\vartheta),\ldots,\theta_{s,j}(\vartheta))$ are available $(j=1,\ldots,m)$.

We consider $m$ null hypotheses and their corresponding alternatives given by ${\theta_{j}(\vartheta)\in H_{j}\;\:\text{versus}\;\;\theta_{j}(\vartheta)\in K_{j}=\Theta_{j}'\setminus H_{j}}$, where $H_{j}$ and $K_{j}$ are non-empty subsets of $\Theta_{j}'$ and Borel sets of $\mathbb{R}^{s}\ (j=1,\ldots, m)$.

Furthermore, we assume that marginal tests $\varphi_{j}$ for testing $H_{j}$ against $K_{j}$ are constructed as $\varphi_{j}(x)=\textbf{1}\{T_{j}(x)\in\Gamma_{j}(\alpha)\}$, where $\Gamma_{j}(\alpha)$ denotes a rejection region, $\alpha\in(0,1)$ denotes a fixed, local significance level, $x\in\Omega$ an observation, and $T_{j}:\Omega\to \mathbb{R}$ a measurable mapping such that the test statistic $T_{j}(X)$ has a continuous cumulative distribution function under any $\vartheta\in\Theta\ {(j=1,\ldots,m)}$. We often write $\hat{\theta}_{j}$ or $T_{j}$ instead of $\hat{\theta}_{j}(X)$ and $T_{j}(X)$, respectively ${(j=1,\ldots,m)}$.

The following general assumptions are made:
\begin{description}
\item[$(GA1)$] For all $j=1,\ldots,m$, there exists a constant $c_{j}\in[0,1]$, such that ${\{x\in\Omega:\;T_{j}(x)\in\Gamma_{j}(c_{j})\}=\{x\in\Omega:\;\hat{\theta}_{j}(x)\in K_{j}\}}$ holds.
\item[$(GA2)$] Nested rejection regions: for every $j=1,\ldots,m$ and $\alpha'<\alpha$, it holds $\Gamma_{j}(\alpha')\subseteq\Gamma_{j}(\alpha)$.
\item[$(GA3)$] For every $j=1,\ldots,m$ and $\alpha\in(0,1)$, it holds $\underset{\vartheta:\theta_{j}(\vartheta)\in H_{j}}{\mathrm{sup}}\text{pr}_{\vartheta}(T_{j}\in\Gamma_{j}(\alpha))=\alpha$. 
\item[$(GA4)$] For every $j=1,\ldots,m$, the set of LFCs for $\varphi_{j}$, i.e. the set of parameters that yield the supremum in $(GA3)$, does not depend on $\alpha$. 
\end{description}
The conditions $(GA2)-(GA4)$ are the same as required for the models in \cite{dickhaus2013randomized}, whereas for assumption $(GA1)$ in \cite{dickhaus2013randomized} only the condition ${\{x\in\Omega:\;T_{j}(x)\in\Gamma_{j}(\alpha)\}\subseteq\{x\in\Omega:\;\hat{\theta}_{j}(x)\in K_{j}\}}$ for $\alpha$ small enough has to be met.

Assumption $(GA1)$ serves as a connection between the test statistic $T_{j}(X)$ and the estimator $\hat{\theta}_{j}(X)$, for $1 \leq j \leq m$. It requires, that $\{\hat{\theta}_{j}\in K_{j}\}$ is in itself a rejection event at level $c_{j}$. Furthermore, assumption $(GA2)$ together with $(GA1)$ implies
\begin{align}
\label{eq:GA1add}
\left.
\begin{array}{l l}
\{x\in\Omega:\;T_{j}(x)\in\Gamma_{j}(\alpha)\}\subseteq\{x\in\Omega:\;\hat{\theta}_{j}(x)\in K_{j}\},\ & \alpha<c_{j}\\ 
\{x\in\Omega:\;T_{j}(x)\in\Gamma_{j}(\alpha)\}\supseteq\{x\in\Omega:\;\hat{\theta}_{j}(x)\in K_{j}\},\ & \alpha>c_{j}
\end{array}
\right.
\end{align}
for all $x\in\Omega\ (j=1,\ldots,m)$. Assumption $(GA3)$ means that under any LFC for $\varphi_{j}$ the rejection probability is exactly $\alpha$.

LFC-based $p$-values for the marginal tests $\varphi_{j}$ are formally defined as
\begin{equation*}
p_{j}^{LFC}(X)=\underset{\{\tilde{\alpha}\in(0,1):T_{j}(x)\in\Gamma_{j}(\tilde{\alpha})\}}{\mathrm{inf}}\;\underset{\{\vartheta:\theta_{j}(\vartheta)\in H_{j}\}}{\mathrm{sup}}\text{pr}_{\vartheta}(T_{j}(X)\in\Gamma_{j}(\tilde{\alpha})).
\end{equation*}

Under assumptions $(GA2)$ -- $(GA4)$, we obtain that
\[p_{j}^{LFC}(X)={\mathrm{inf}\{\tilde{\alpha}\in(0,1):T_{j}(X)\in\Gamma_{j}(\tilde{\alpha})\}}\ (j=1,\ldots,m).
\] 
Such LFC-based $p$-values $p_{j}^{LFC}(X)$ are uniformly distributed on $[0,1]$ under any LFC for $\varphi_{j}$ \citep[Lemma $3.3.1$]{lehmann2006testing}. If $\Gamma_{j}(\alpha)=(F_{j}^{-1}(1-\alpha),\infty)$, where $F_{j}$ is the cumulative distribution function of $T_{j}(X)$ under an LFC for $\varphi_{j}$, the above definition leads to $p_{j}^{LFC}(X)=1-F_{j}(T_{j}(X))$.

\begin{example}
\label{ex:generalassumptions} 
Models $1$ and $2$ in \cite{dickhaus2013randomized} are one-sided normal means models that fulfill the general assumptions $(GA1)$ -- $(GA4)$. Notice that the indices $i$ in \cite{dickhaus2013randomized} correspond to the indices $j$ in our notation, and that the dimension $s$ of the derived parameters is one in both models.

\cite{dickhaus2013randomized} showed that the general assumptions $(GA2)$ -- $(GA4)$ hold. Our stricter assumption $(GA1)$ follows in both models from the fact that the estimators $\hat{\theta}_{j}(X)$ are positive, i.e. inside the alternative, if and only if the test statistics $T_{j}(X)$ are positive, which is equivalent to $T_{j}(X)\in\Gamma_{j}(1/2)$, such that $c_{j}=1/2$ in $(GA1)\ (j=1,\ldots, m)$.
\end{example}


\section{\texorpdfstring{The randomized $p$-values}{The randomized p-values}}
\label{sec:randomizedpvalues}
\subsection{General properties}
\label{sec:generalproperties}

Let $U_{1},\ldots,U_{m}$ be stochastically independent and identically, uniformly distributed on $[0,1]$, such that each $U_{j}$ is stochastically independent of $X$. We obtain randomized $p$-values $p_{j}^{rand}$ by mixing $U_{j}$ and $p_{j}^{LFC}$ in a data-dependent manner, specifically 
\begin{equation*}
p_{j}^{rand}(X,U_{j})=w_{j}(X)~U_{j}+(1-w_{j}(X))~G_{j}\big(p_{j}^{LFC}(X)\big),
\end{equation*}
where $G_{j}$ is a suitable function necessary for the validity of the randomized $p$-values and ${0\leq w_{j}(x)\leq 1}$ are data-dependent weights $(j=1,\ldots,m)$.

We consider the choice $w_{j}(X)=\textbf{1}_{H_{j}}\{\hat{\theta}_{j}(X)\}\ (j=1,\ldots,m)$, which follows the definition of randomized $p$-values as introduced in \cite{dickhaus2013randomized}.
\begin{definition}
\label{def:randomized}\ \\
We define randomized $p$-values as follows
\begin{equation*}
p_{j}^{rand}(X,U_{j})=U_{j}~\textbf{1}_{H_{j}}\{\hat{\theta}_{j}(X)\}+G_{j}\big(p_{j}^{LFC}(X)\big)~\textbf{1}_{K_{j}}\{\hat{\theta}_{j}(X)\},
\end{equation*}
where $G_{j}$ denotes the conditional cumulative distribution function of $p_{j}^{LFC}(X)$ given the event $\{\hat{\theta}_{j}\in K_{j}\}$ under any LFC for $\varphi_{j}\ (j=1,\ldots,m)$. 
\end{definition}
Ideally, we want $p$-values to be uniformly distributed on $[0,1]$ under null hypotheses. For a fixed $j=1,\ldots,m$, we therefore set $p_{j}^{rand}(X)=U_{j}$ if $\hat{\theta}_{j}(X)\in H_{j}$ holds. Due to \eqref{eq:GA1add}, $\varphi_{j}(x)=0$ whenever $\hat{\theta}_{j}(x)\in H_{j}$ holds, when applying a local significance level $\alpha<c_{j}$. This means that in case of $\hat{\theta}_{j}(x)\in H_{j}$ we cannot reject $H_{j}$ at a significance level lower than $c_{j}$. Since $c_{j}$ can be very large, e.\ g. $1/2$ in Example~\ref{ex:generalassumptions} and even larger in our models for replicability analysis in Section~\ref{sec:replicabilityanalysis}, we can, in practice, assume that $H_{j}$ is true in case of $\hat{\theta}_{j}(X)\in H_{j}$, and switch to a uniform variate $U_{j}$ that has the desired properties for a $p$-value under $H_{j}$.

In the following theorem we give formulas for the calculation of the function $G_{j}$ and the randomized $p$-value $p_{j}^{rand}\ (j=1,\ldots,m)$.

\begin{theorem}
\label{thm:1}\ \\
Let $j\in\{1,\ldots,m\}$ be fixed and $\vartheta_{0}\in\Theta$ with $\theta_{j}(\vartheta_{0})\in H_{j}$ be any LFC for $\varphi_{j}$. Under assumptions $(GA1)$ -- $(GA4)$ we obtain the following.
\begin{description}
\item[$1.$] It holds that $c_{j}=\text{pr}_{\vartheta_{0}}(\hat{\theta}_{j}(X)\in K_{j})$.
\item[$2.$] The conditional cumulative distribution function $G_{j}$ of $p_{j}^{LFC}(X)$ given $\hat{\theta}_{j}\in K_{j}$ is a piecewise linear function in $t\in[0,1]$, more precisely it holds $G_{j}(t)= t \textbf{1}_{[0,c_{j}]}(t) / c_j +\textbf{1}_{(c_{j},1]}(t)$.
\item[$3.$] The randomized $p$-values, as defined in Definition~\ref{def:randomized}, are of the form
\begin{equation*}
p_{j}^{rand}(X,U_{j})=U_{j}~\textbf{1}_{H_{j}}\{\hat{\theta}_{j}(X)\}+p_{j}^{LFC}(X)\,c_{j}^{-1}\,\textbf{1}_{K_{j}}\{\hat{\theta}_{j}(X)\}.
\end{equation*}
\end{description}
\end{theorem}

Since $p_{j}^{LFC}(x)<c_{j}$ implies $\hat{\theta}_{j}(x)\in K_{j}$, and $p_{j}^{LFC}(x)>c_{j}$ implies $\hat{\theta}_{j}(x)\in H_{j}$, for all $x\in\Omega$, we have
\begin{equation*}
p_{j}^{rand}(x,u_{j})=u_{j}\,\textbf{1}_{(c_{j},1]}\{p_{j}^{LFC}(x)\}+p_{j}^{LFC}(x)\,c_{j}^{-1}\,\textbf{1}_{[0,c_{j})}\{p_{j}^{LFC}(x)\}
\end{equation*}
for any $x\in\Omega$ and $u_{j}\in[0,1]$, when disregarding the case $p_{j}^{LFC}(x)=c_{j}$, for which $p_{j}^{rand}(x,u_{j})$ is either $1$ or $u_{j}\ (j=1,\ldots,m)$.

\begin{example}
\label{ex:randomizedpvalues}
We apply Theorem~\ref{thm:1} to both models in Example~\ref{ex:generalassumptions}. In both models it holds, that $p_{j}^{LFC}(x)<t$ is equivalent to $T_{j}(x)\in\Gamma_{j}(t)$ for all $x\in\Omega$ and $t\in[0,1]$. In particular, $p_{j}^{LFC}(x)<c_{j}$ is equivalent to $\hat{\theta}_{j}(x)\in K_{j},\;x\in\Omega$, such that $\textbf{1}_{H_{j}}\{\hat{\theta}_{j}(X)\}$ and $\textbf{1}_{K_{j}}\{\hat{\theta}_{j}(X)\}$ in Part $3$ of Theorem~\ref{thm:1} can be replaced by $\textbf{1}_{[c_{j},1]}\{p_{j}^{LFC}(X)\}$ and $\textbf{1}_{[0,c_{j})}\{p_{j}^{LFC}(X)\}$, respectively. Let $j=1,\ldots,m$ be fixed.
\begin{description}
\item[$1.$] (Multiple Z-tests model) From Theorem $1$ it follows that $c_{j}=\text{pr}_{\vartheta_{0}}(\hat{\theta}_{j}\in K_{j})=1/2$, and 
\begin{align*}
G_{j}(t)&=2\,t\,\textbf{1}_{[0,\frac{1}{2}]}(t)+\textbf{1}_{(\frac{1}{2},1]}(t),& &t\in[0,1]\\
p_{j}^{rand}(x,u_{j})&=u_{j}\textbf{1}_{(\frac{1}{2},1]}\big\{p_{j}^{LFC}(x)\big\}+2\,p_{j}^{LFC}(x)\textbf{1}_{[0,\frac{1}{2}]}\big\{p_{j}^{LFC}(x)\big\},& &x\in\Omega,u_{j}\in[0,1].
\end{align*}
\item[$2.$] (Multiple $t$-tests model) Analogously to the multiple Z-tests model, it follows that $c_{j}=1/2$ and
\begin{align*}
G_{j}(t)&=2\,t\,\textbf{1}_{[0,\frac{1}{2}]}(t)+\textbf{1}_{(\frac{1}{2},1]}(t),& &t\in[0,1]\\
p_{j}^{rand}(x,u_{j})&=u_{j}\textbf{1}_{(\frac{1}{2},1]}\big\{p_{j}^{LFC}(x)\big\}+2\,p_{j}^{LFC}(x)\textbf{1}_{[0,\frac{1}{2}]}\big\{p_{j}^{LFC}(x)\big\},& &x\in\Omega,u_{j}\in[0,1],
\end{align*}
directly from Theorem~\ref{thm:1}.
\end{description}
These results agree with the calculations in \cite[pp.1971, 1973]{dickhaus2013randomized}.
\end{example}

\subsection{\texorpdfstring{Conditions for the validity of the randomized $p$-values}{Conditions for the validity of the randomized p-values}}
\label{sec:conditionsvalidity}
As mentioned before, valid $p$-values are usually required for a conservative estimation of the proportion $\pi_{0}$ of true null hypotheses, particularly if the Schweder-Spj{\o}tvoll estimator $\hat{\pi}_{0}$ is applied. This section provides some conditions for the validity of the randomized $p$-values as defined in Definition~\ref{def:randomized} for our model setup.
\begin{theorem}
\label{thm:2}\ \\
Let $j\in\{1,\ldots,m\}$ be fixed. Under the general assumptions $(GA1)$ -- $(GA4)$, assume that $p_{j}^{LFC}(X)$ has a continuous and strictly increasing cumulative distribution function under any $\vartheta\in\Theta$. Then, the randomized $p$-value $p_{j}^{rand}$, as defined in Definition~\ref{def:randomized}, is a valid $p$-value if and only if 
\begin{equation*}
\text{pr}_{\vartheta}\big(T_{j}(X)\in\Gamma_{j}(z)\big)\leq z\,\frac{\text{pr}_{\vartheta}(\hat{\theta}_{j}\in K_{j})}{\text{pr}_{\vartheta_{0}}(\hat{\theta}_{j}\in K_{j})},\quad 0\leq z\leq \text{pr}_{\vartheta_{0}}(\hat{\theta}_{j}\in K_{j}),
\end{equation*}
for any $\vartheta\in\Theta$ with $\theta_{j}(\vartheta)\in H_{j}$ and for any LFC $\vartheta_{0}$ for $\varphi_{j}$.
\end{theorem}

In many applications, a rejection of $H_{j}$ after observing $T_{j}(x)$ implies a rejection of $H_{j}$ if we observe larger test values $T_{j}(y)\geq T_{j}(x)$, $x,y\in\Omega$. More specifically, the rejection regions are often of the form $\Gamma_{j}(\alpha)=(b(\alpha),\infty)$ for some non-decreasing boundary function $b:[0,1]\to \mathbb{R}$. Usually, $b(\alpha)=F^{-1}(1-\alpha),\;\alpha\in[0,1]$, where $F$ is the cumulative distribution function of $T_{j}(X)$ under an LFC for $\varphi_{j}$, such that $(GA3)$ holds. 
Among others, the models from Example~\ref{ex:generalassumptions} fulfill this condition, under which the validity of the randomized $p$-value $p_{j}^{rand}$ follows from $T_{j}(X)$ being smaller in the hazard rate order under any $\vartheta\in\Theta$ with $\theta_{j}(\vartheta)\in H_{j}$ than under an LFC for $\varphi_{j}$.

We denote the hazard rate order and the likelihood ratio order with $"\leq_{\mathrm{hr}}"$ and $"\leq_{\mathrm{lr}}"$, respectively. For a more detailed introduction to our notations we refer to the appendix.

\begin{theorem}
\label{thm:3}
Let a model as in Section~\ref{sec:modelsetup} be given and $j=1,\ldots,m$ be fixed. We assume that the rejection regions are of the form $\Gamma_{j}(\alpha)=(F^{-1}(1-\alpha),\infty),~\alpha\in[0,1]$, where $F$ is the cumulative distribution function of $T_{j}(X)$ under any LFC $\vartheta_{0}\in\Theta$ for $\varphi_{j}$.

Then the randomized $p$-value $p_{j}^{rand}$ as defined in Definition~\ref{def:randomized} is valid if it holds $T_{j}(X)^{(\vartheta)}\leq_{\mathrm{hr}}T_{j}(X)^{(\vartheta_{0})}$ for all $\vartheta\in\Theta$ with $\theta_{j}(\vartheta)\in H_{j}$ and any LFC $\vartheta_{0}$ for $\varphi_{j}$.
\end{theorem}

\begin{corollary}
\label{cor:thm3corollary}\ \\
By Theorem $1.C.2$ in \cite{shaked2007stochastic}, replacing the hazard rate order by the likelihood ratio order in Theorem~\ref{thm:3} is also sufficient for the validity of $p_{j}^{rand}$.
\end{corollary}

\begin{example}
\label{ex:thm3}
We show via Theorem~\ref{thm:3} that the randomized $p$-values as calculated in Example~\ref{ex:randomizedpvalues} are valid. Let $j\in\{1,\ldots,m\}$ be fixed.
\begin{description}
\item[$1.$] (Multiple Z-tests Model) Let $\vartheta_{0}\in\Theta$ with $\theta_{j}(\vartheta_{0})=0$ be an LFC for $\varphi_{j}$, and $\vartheta\in\Theta$ with $\theta_{j}(\vartheta)=\mu_{j}<0$. Recall that $T_{j}(X)=\bar{X}_{j}$ is normally distributed on $\mathbb{R}$ with variance $1/n_{j}$ and expected values $\mu_{j}$ and $0$ under $\vartheta$ and $\vartheta_{0}$, respectively. It is easy to show that $f_{\vartheta_{0}}(t)/f_{\vartheta}(t)$ is non-decreasing in $t$ and therefore $T_{j}(X)^{(\vartheta)}\leq_{\mathrm{lr}}T_{j}(X)^{(\vartheta_{0})}$ holds, where $f_{\vartheta}$ and $f_{\vartheta_{0}}$ denote the Lebesgue densities of $N(\mu_{j},1/n_{j})$ or $N(0,1/n_{j})$, respectively. According to Corollary~\ref{cor:thm3corollary} our randomized $p$-values $p_{j}^{rand}$ are valid in this model.

\item[$2.$] (Multiple $t$-tests Model) Now we have that $T_{j}(X)=n_{j}^{1/2}\bar{X}_{j}/S_{j}$ possesses a non-central $t$-distribution with non-centrality parameter $\tau_{j}(\vartheta)$ and $n_{j}-1$ degrees of freedom, $\tau_{j}(\vartheta)=n_{j}^{1/2}\mu_{j}/\sigma_{j}$, and $\mu_{j}=\theta_{j}(\vartheta)$.

According to \cite[p. 639]{karlin1956distributions} and \cite[p. 126]{karlin1956decision}, non-central $t$-distributions $\big(t_{\mu,\nu}\big)_{\mu\in \mathbb{R}}$ have monotone likelihood ratio, i.e. $t_{\mu_{1},\nu}\leq_{\mathrm{lr}}t_{\mu_{2},\nu}$ if and only if $\mu_{1}\leq\mu_{2}$. For $\vartheta,\vartheta_{0}\in\Theta$ with $\theta_{j}(\vartheta)\leq 0$ and $\theta_{j}(\vartheta_{0})=0$, it is $\tau_{j}(\vartheta)=n_{j}^{1/2}\theta_{j}(\vartheta)/\sigma_{j}\leq 0=\tau_{j}(\vartheta_{0})$, and therefore $T_{j}(X)^{(\vartheta)}\leq_{\mathrm{lr}}T_{j}(X)^{(\vartheta_{0})}$. According to Corollary~\ref{cor:thm3corollary} our randomized $p$-values $p_{j}^{rand}$ in this model are valid.
\end{description}
\end{example}

Under certain conditions randomized $p$-values $p_{j}^{rand}$ as defined in Definition$~\ref{def:randomized}$ are closer to $\mathrm{Uni}[0,1]$ than their LFC-based counterparts $p_{j}^{LFC}$ under the null hypothesis $H_{j}$, that is,
\[ 
\mathrm{Uni}[0,1]\leq_{\mathrm{st}}p_{j}^{rand}(X,U_{j})^{(\vartheta)}\leq_{\mathrm{st}}p_{j}^{LFC}(X,U_{j})^{(\vartheta)}
\]
or, equivalently, 
\[
\text{pr}_{\vartheta}(p_{j}^{LFC}(X,U_{j})\leq t)\leq\text{pr}_{\vartheta}(p_{j}^{rand}(X,U_{j})\leq t)\leq t
\]
for all $t\in[0,1]$ and $\vartheta\in\Theta$ with $\theta_{j}(\vartheta)\in H_{j}$.

\begin{theorem}
\label{thm:stochasticorder}
Let a model as in Section$~\ref{sec:modelsetup}$ be given and $j\in\{1,\ldots,m\}$ be fixed. If the cumulative distribution function of $p_{j}^{LFC}(X)$ is convex under $\vartheta\in\Theta$, then it holds 
\[ 
\mathrm{Uni}[0,1]\leq_{\mathrm{st}}p_{j}^{rand}(X,U_{j})^{(\vartheta)}\leq_{\mathrm{st}}p_{j}^{LFC}(X,U_{j})^{(\vartheta)}.
\] 
On the other hand, if the cumulative distribution function of $p_{j}^{LFC}(X)$ is concave under $\vartheta\in\Theta$, then it holds 
\[ 
p_{j}^{LFC}(X,U_{j})^{(\vartheta)}\leq_{\mathrm{st}}p_{j}^{rand}(X,U_{j})^{(\vartheta)}\leq_{\mathrm{st}}\mathrm{Uni}[0,1].
\] 
\end{theorem}

\begin{remark}
\begin{description}
\item[$1.$] Under the null hypothesis $H_{j}$ the cumulative distribution function of $p_{j}^{LFC}(X)$ can never be concave.
\item[$2.$] From Theorem$~\ref{thm:stochasticorder}$, if the cumulative distribution function of $p_{j}^{LFC}(X)$ is convex under all $\vartheta\in\Theta$ with $\theta_{j}(\vartheta)\in H_{j}$, the randomized $p$-value $p_{j}^{rand}$ is a valid $p$-value.
\item[$3.$] If the rejection regions are of the form $\Gamma_{j}(\alpha)=(F^{-1}(1-\alpha),\infty)$, where $F$ is the cumulative distribution function of $T_{j}(X)$ under an LFC for $\varphi_{j}$ (cf. Theorem$~\ref{thm:3}$), the condition mentioned in the second remark is stronger than the condition in Theorem$~\ref{thm:3}$. Namely, the convexity of the cumulative distribution function of $p_{j}^{LFC}(X)$ under $\vartheta$ is equivalent to $T_{j}(X)^{(\vartheta)}\leq_{\mathrm{lr}}T_{j}(X)^{(\vartheta_{0})}$ whenever $\vartheta_{0}$ is an LFC.
\end{description}
\end{remark}

\begin{example}
For both models from our ongoing examples the cumulative distribution function of $p_{j}^{LFC}$ is convex under $H_{j}$ and concave under $K_{j}$. Therefore both conditions in Theorem$~\ref{thm:stochasticorder}$ are satisfied and $p_{j}^{rand}$ is always closer to $\mathrm{Uni}[0,1]$ than $p_{j}^{LFC}$ (in the sense of stochastic order).  
\end{example}

\section{\texorpdfstring{Randomized $p$-values in replicability analysis}{Randomized p-values in replicability analysis}} \label{sec:replicabilityanalysis}
\subsection{Model setup}
\label{sec:modelreplicability}

We come back to the framework of bio-marker identification. We want to find bio-markers that have been verified in at least $\gamma$ studies, where the parameter $\gamma\in\{2,\ldots,s\}$ is pre-defined and fixed. For $\gamma=s$, we declare discoveries replicated, only if they have been made in each considered study. Also, it is clear, that the set of false null hypotheses is non-increasing in $\gamma$.

For each endpoint $j$ and study $i$, we denote the true effect on the considered disease state by a parameter $\theta_{i,j}$, where $\theta_{i,j}>0$ mean positive effects. We consider an endpoint to be a bio-marker only if it exhibits a positive effect on the disease. This can be replaced by testing for any one fixed, directional association between the endpoint and the disease. The parameters $\theta_{i,j}$ may differ inherently in $i$ due to the different settings across the studies like different populations or different laboratory / statistical methods.

We consider the model from Section~\ref{sec:modelsetup} for $s\geq 2$. Unless stated otherwise, we only consider $\Theta'= \mathbb{R}^{sm}$, i.e. each derived parameter $\theta_{i,j}$ may take any value in $\mathbb{R}\ (i=1,\ldots,s;\ j=1,\ldots,m)$.

Before we get to constructing the test statistics $T_{j}(X)$ and the rejection regions $\Gamma_{j}(\alpha)$, we first make some requirements about the marginal model setup. This will make it easier to present sufficient conditions for the general assumptions $(GA1)$ -- $(GA4)$ from Section~\ref{sec:modelsetup}. We do not require the data for different endpoints in the same study to be independent.

For every study $i=1,\ldots,s$ and marker $j=1,\ldots,m$ we test for a positive effect size of endpoint $j$ with regard to the disease, $H_{i,j}=\{\theta_{i,j}\leq 0\}$ vs. $K_{i,j}=\{\theta_{i,j}>0\}$. We assume that a consistent and, at least asymptotically, unbiased estimator $\hat{\theta}_{i,j}:\Omega\to \mathbb{R}$ for $\theta_{i,j}(\vartheta)$ is available. Furthermore, the marginal test $\varphi_{i,j}$ for testing $H_{i,j}$ against $K_{i,j}$ is based on a test statistic $T_{i,j}(X)$ and rejection regions $\Gamma_{i,j}(\alpha)$, where $\alpha\in(0,1)$ denotes the (local) significance level, $x\in\Omega$ an observation, and $T_{i,j}:\Omega\to \mathbb{R}$ a measurable mapping such that the test statistic $T_{i,j}(X)$ has a continuous cumulative distribution function under any $\vartheta\in\Theta$. The corresponding LFC-based $p$-values are then denoted by $p_{i,j}(X)$.

For every $i=1,\ldots,s$ and $j=1,\ldots,m$ we make the following assumptions:
\begin{description}
\item[$(RA1)$] It holds $\Gamma_{i,j}(\alpha)=(F_{i,j}^{-1}(1-\alpha),\infty)$ and $p_{i,j}(X)=1-F_{i,j}\big(T_{i,j}(X)\big)$, the set of LFCs for $\varphi_{i,j}$ is $\{\vartheta\in\Theta:\theta_{i,j}(\vartheta)=0\}$, and $F_{i,j}$ denotes the cumulative distribution function of $T_{i,j}(X)$ under an LFC for $\varphi_{i,j}$.
\item[$(RA2)$] The assumptions $(GA1)-(GA4)$ are fulfilled. We denote with $c_{i,j}$ the value, that satisfies $\{x\in\Omega:T_{i,j}(x)\in\Gamma_{i,j}(c_{i,j})\}=\{x\in\Omega:\hat{\theta}_{i,j}(x)\in K_{i,j}\}=\{x\in\Omega:\hat{\theta}_{i,j}(x)>0\}$ for assumption $(GA1)$.
\item[$(RA3)$] There exists a $d_{j}\in(0,1)$ such that $p_{i,j}(x)<d_{j}$ if and only if $\hat{\theta}_{i,j}(x)>0$, for all $x\in\Omega$ and $1 \leq i \leq s$.
\item[$(RA4)$] It holds $\underset{\theta_{i,j}(\vartheta)\rightarrow\infty}{\mathrm{lim}}\text{pr}_{\vartheta}\big(F_{i,j}(T_{i,j}(X))=1\big)=1$.
\end{description}
In one-sided problems, assumption $(RA1)$ is usually fulfilled. Due to $(RA1)$ it now holds 
\begin{equation}
\label{eq:newLFCGamma}
p_{i,j}(x)<t\Longleftrightarrow T_{i,j}(x)\in\Gamma_{i,j}(t),\;x\in\Omega.
\end{equation} 
Assumption $(RA3)$ is akin to assumption $(GA1)$ from Section~\ref{sec:modelsetup}, and follows from $(RA2)$ if and only if $d_{j}=c_{1,j}=\cdots=c_{s,j}$ holds $(j=1,\ldots,m)$.

For convenience we write

\begin{equation*}
\underset{\theta_{i,j}(\vartheta)\rightarrow\infty}{\mathrm{lim}}\text{pr}_{\vartheta}(T_{i,j}(X)\leq t)=\text{pr}_{\vartheta_{1}}(T_{i,j}(X)\leq t),\;t\in \mathbb{R},
\end{equation*} 

where $\vartheta_{1}$ is such, that $\theta_{i,j}(\vartheta_{1})=\infty$, although $\vartheta_{1}$ is technically not a parameter. Assumption $(RA4)$ is equivalent to $p_{i,j}(X)$ being zero almost surely under any such $\vartheta_{1}\ (i=1,\ldots,s;\ j=1,\ldots,m)$.

For any endpoint we define replicability of a bio-marker finding as the evidence of a positive effect size in at least $\gamma$ out of the $s$ studies. Let $H_{1},\ldots,H_{m}$ be the \textit{non-replicability} null hypotheses and $K_{1},\ldots,K_{m}$ be the respective alternative hypotheses. Formally, we define 
\begin{align*}
H_{j}&=\{(\theta_{1,j},\ldots,\theta_{s,j})\in\Theta_{j}'\mid \theta_{i,j}\leq 0\;\text{for at least $s-\gamma+1$ indices $i\in\{1,\ldots,s\}$}\},\\
K_{j}&=\{(\theta_{1,j},\ldots,\theta_{s,j})\in\Theta_{j}'\mid \theta_{i,j}>0\;\text{for at least $\gamma$ indices $i\in\{1,\ldots,s\}$}\}
\end{align*}
for $j=1,\ldots,m$.
Furthermore, we define consistent and, at least asymptotically, unbiased estimators for $\theta_{j}(\vartheta)=(\theta_{1,j}(\vartheta),\ldots,\theta_{s,j}(\vartheta))$ by $\hat{\theta}_{j}=(\hat{\theta}_{1,j},\ldots,\hat{\theta}_{s,j}):\Omega\to \mathbb{R}^{s}\ (j=1,\ldots,m)$.

To make a decision about the replicability of an effect for marker $j$ we consider the ordered p-values $p_{(1),j}<\cdots<p_{(s),j}$ for the hypotheses $H_{i,j}\ (i=1,\ldots,s)$, in the $s$ studies. One plausible  approach is to look at the $\gamma$ smallest p-values and reject $H_{j}$ if these are all below a suitable threshold. We therefore define $T_{j}(X)=1-p_{(\gamma),j}(X)$ and $\Gamma_{j}(\alpha)=(F_{\mathrm{Beta(s-\gamma+1,1)}}^{-1}(1-\alpha),1]$, thus rejecting $H_{j}$ if the $\gamma$ smallest $p$-values are all below $1-F_{\mathrm{Beta(s-\gamma+1,1)}}^{-1}(1-\alpha)$, where $F_{\mathrm{Beta(s-\gamma+1,1)}}$ denotes the cumulative distribution function of the $\mathrm{Beta}(s-\gamma+1,1)$ distribution. For the LFC-based $p$-values we then have $p_{j}^{LFC}(x)=1-F_{\mathrm{Beta}(s-\gamma+1,1)}\big(T_{j}(x)\big)\ (j=1,\ldots,m)$.

Let $G_{j}$ be the conditional cumulative distribution function of $p_{j}^{LFC}(X)$ given $\hat{\theta}_{j}(X)\in K_{j}$ under any LFC $\vartheta_{0}\in\Theta$ for $\varphi_{j}$, cf. Definition~\ref{def:randomized}.  According to Theorem \ref{thm:1} with $c_{j}=1-(1-d_{j})^{n-\gamma+1}$, it holds that
\begin{equation*}
G_{j}(t)=\frac{t}{1-(1-d_{j})^{n-\gamma+1}}\textbf{1}_{[0,1-(1-d_{j})^{n-\gamma+1}]}(t)+\textbf{1}_{(1-(1-d_{j})^{n-\gamma+1},1]}(t),\quad 0\leq t\leq 1,
\end{equation*}
and 
\begin{align*}
p_{j}^{rand}(x,u_{j})=&u_{j}\textbf{1}_{[1-(1-d_{j})^{n-\gamma+1},1]}\big\{p_{j}^{LFC}(x)\big\}\\
&+\frac{p_{j}^{LFC}(x)}{1-(1-d_{j})^{n-\gamma+1}}\textbf{1}_{[0,1-(1-d_{j})^{n-\gamma+1})}\big\{p_{j}^{LFC}(x)\big\}
\end{align*}
for $x\in\Omega$ and  $0\leq u_{j}\leq 1$; see the proof of Lemma~\ref{lm:replicabilitygeneralassumptions} in \ref{sup:sec:proofs}.

\begin{lemma}
\label{lm:replicabilitygeneralassumptions}
If assumptions $(RA1)$ -- $(RA4)$ are fulfilled, the model in Section~\ref{sec:modelreplicability} satisfies the general assumptions $(GA1)$ -- $(GA4)$ from Section~\ref{sec:modelsetup}.
\end{lemma}

Lemma~\ref{lm:replicabilitygeneralassumptions} allows us to check the general assumptions $(GA1)$ -- $(GA4)$ of the overall model by looking at the single studies. As such, it is not difficult to provide models that fulfill $(GA1)$ -- $(GA4)$.
\begin{example}
\label{ex:replicabilitymodels}
In the following, we consider models, in which we utilize either a $Z$-test or a $t$-test for each study $i$ and endpoint $j$.
\begin{description}
\item[] Model $1$: In each study $i=1,\ldots,s$ we consider a multiple $Z$-tests model. For fixed sample sizes $n_{i,j},\ (i=1,\ldots,s;\ j=1,\ldots,m)$, we consider the observations $x\in\Omega= \mathbb{R}^{\sum_{i,j}n_{i,j}}$ as realizations of $X=\{X_{k}^{(i,j)}:i=1,\ldots,s, j=1,\ldots,m, k=1,\ldots,n_{i,j}\}$.

For each study $i$ and marker $j$ the observations $X_{1}^{(i,j)},\ldots,X_{n_{i,j}}^{(i,j)}$ are stochastically independent and identically, normally distributed on $\mathbb{R}$ with expected value $\theta_{i,j}(\vartheta)$ and variance $1$, where $\vartheta\in\Theta$ is the underlying parameter. It is $\Theta= \mathbb{R}^{sm}$, where we denote the parameters by $\vartheta=(\mu_{i,j}: 1\leq i\leq s,1\leq j\leq m)$, such that $\theta_{i,j}(\vartheta)=\mu_{i,j}$.

As before, we test the null hypotheses $H_{i,j}=\{\mu_{i,j}\leq 0\}$ against the alternatives $K_{i,j}=\{\mu_{i,j}>0\}$. A consistent and unbiased estimator for $\mu_{i,j}$ is $\hat{\theta}_{i,j}(X)=\bar{X}_{i,j}= n_{i,j}^{-1} \sum_{k=1}^{n_{i,j}}X_{k}^{(i,j)}$, which is normally distributed on $\mathbb{R}$ with expected value $\mu_{i,j}$ and variance $1/n_{i,j}$.

Furthermore, we choose test statistics $T_{i,j}(X)=\hat{\theta}_{i,j}(X)$ and rejection regions $\Gamma_{i,j}(\alpha)=\big(\Phi_{(0,1/n_{i,j})}^{-1}(1-\alpha),\infty\big)$, where $\Phi_{(\mu,\sigma^{2})}$ is the cumulative distribution function of the normal distribution on $\mathbb{R}$ with expected value $\mu$ and variance $\sigma^{2}$.

Assumptions $(RA1)$ -- $(RA3)$ have already been discussed before, with $c_{i,j}=d_{j}=1/2$ for all $i,j$, and $(RA4)$ is clear. Under this model, due to Lemma~\ref{lm:replicabilitygeneralassumptions}, assumptions $(GA1)$ -- $(GA4)$ are fulfilled.
\item[] Model $2$: For multiple $t$-tests instead of $Z$-tests, where the observations have unknown variance (cf. Model $2$ in \cite{dickhaus2013randomized}), assumptions $(RA1)$ -- $(RA4)$ are analogous to verify, which again results in an overall model that fulfills assumptions $(GA1)$ -- $(GA4)$.
\end{description}
\end{example}

We give a sufficient condition based on Theorem~\ref{thm:3} for the validity of the randomized $p$-values $p_{j}^{rand}$, that result from our model setup.
\begin{theorem}
\label{thm:4}
Let a model as above be given, such that assumptions $(RA1)$ -- $(RA4)$ are fulfilled, and let $j\in\{1,\ldots,m\}$ be fixed.

If, for all $i=1,\ldots,s$ and $\vartheta,\vartheta_{0}\in\Theta$ with $\theta_{j}(\vartheta),\theta_{j}(\vartheta_{0})\in H_{j}$ and $\theta_{i,j}(\vartheta)\leq 0=\theta_{i,j}(\vartheta_{0})$, it holds $T_{i,j}(X)^{(\vartheta)}\leq_{\mathrm{hr}}T_{i,j}(X)^{(\vartheta_{0})}$, then $p_{j}^{rand}$ is a valid $p$-value.
\end{theorem}

\begin{remark}
\begin{description}
\item[] Theorem$~\ref{thm:4}$ still holds if we replace the hazard rate order $\leq_{\mathrm{hr}}$ by the likelihood ratio order $\leq_{\mathrm{lr}}$.
\item[] Under a model that fulfills Theorem~\ref{thm:4} the randomized $p$-values $p_{i,j}^{rand}\ (i=1,\ldots,s)$, resulting from study $i$ and marker $j$ are valid as well, as a result of Theorem~\ref{thm:3}. 
\end{description}
\end{remark}

\begin{example}
\label{ex:thm4}\ \\
The randomized $p$-values $p_{j}^{rand}\ (j=1,\ldots,m)$, in Models $1$ and $2$, as introduced in 
Example~\ref{ex:replicabilitymodels} are valid. Here, we show that for Model$~1$.

Recall that $T_{i,j}(X)=\hat{\theta}_{i,j}(X)$ is normally distributed on $\mathbb{R}$ with expected value $\theta_{i,j}(\vartheta)$ and variance $1/n_{i,j}$ under $\vartheta\in\Theta$, where $n_{i,j}\ (i=1,\ldots,s;\ j=1,\ldots,m)$, are the fixed sample sizes. For $i\in\{1,\ldots,s\}$ and $\vartheta,\vartheta_{0}\in\Theta$, such that $\theta_{j}(\vartheta),\theta_{j}(\vartheta_{0})\in H_{j}$ and $\theta_{i,j}(\vartheta)\leq 0$, $\theta_{i,j}(\vartheta_{0})=0$, it holds $T_{i,j}(X)^{(\vartheta)}\leq_{\mathrm{lr}}T_{i,j}(X)^{(\vartheta_{0})}$, cf. Example$~\ref{ex:randomizedpvalues}$. It follows from Theorem$~\ref{thm:4}$, that $p_{j}^{rand}$ is a valid $p$-value $(j=1,\ldots,m)$.

In Fig.$~\ref{fig:cdfs}$ we compare the cumulative distribution functions of $p_{j}^{LFC}$ and $p_{j}^{rand}$ for $\theta_{j}(\vartheta)\in H_{j}$, $\theta_{j}(\vartheta)=(-1.5\,n_{1,j}^{-1/2},\ldots,-1.5\,n_{s-\gamma+1,j}^{-1/2},1,\ldots,1)$, and $\theta_{j}(\vartheta)\in K_{j}$, $\theta_{j}(\vartheta)=(2\,n_{1,j}^{-1/2},\ldots,2\,n_{s,j}^{-1/2})$, in the first and second graph, respectively, where we set $s=10,\gamma=6$, and the sample sizes to $n_{1,j}=\cdots=n_{s,j}=50$.

The left graph shows that the randomized $p$-value $p_{j}^{rand}(X,U_{j})$ is stochastically not larger than the LFC-based $p$-value $p_{j}^{LFC}(X)$ but remains valid, i.e. not smaller than a uniform distribution on $[0,1]$. It is apparent that $p_{j}^{rand}(X,U_{j})$ comes much closer to the uniform distribution on $[0,1]$. The right graph, however, illustrates that the randomized $p$-value $p_{j}^{rand}(X,U_{j})$ is stochastically larger than the LFC-based $p$-value $p_{j}^{LFC}(X)$, under a parameter $\vartheta\in\Theta$ with $\theta_{j}(\vartheta)\in K_{j}$.
\begin{figure}
\includegraphics[width=0.9 \textwidth]{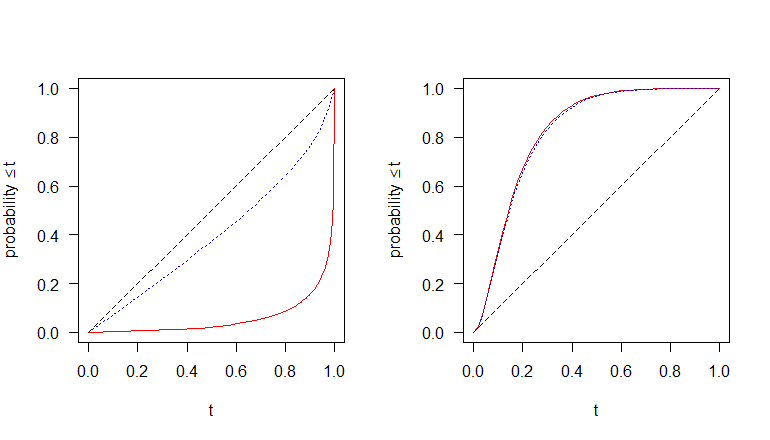}
\caption{A comparison of the cumulative distribution functions of $p_{j}^{LFC}(X)$ (solid) and $p_{j}^{rand}(X,U)$ (dotted) under Model $1$ for $s=10,\gamma=6$, and $n_{1,j}=\cdots=n_{s,j}=50$. The true parameters are $\theta_{j}(\vartheta)=(-1.5\,n_{1,j}^{-1/2},\ldots,-1.5\,n_{s-\gamma+1,j}^{-1/2},1,\ldots,1)$ on the left and $\theta_{j}(\vartheta)=(2\,n_{1,j}^{-1/2},\ldots,2\,n_{s,j}^{-1/2})$ on the right side. For comparison, the dashed lines depict the cumulative distribution function of the standard uniform distribution, which is given by $t \mapsto t$ for $t \in [0, 1]$.}
\label{fig:cdfs}
\end{figure}
\end{example}

\section{Estimation of the proportion of true null hypotheses} \label{sec:estimationtruenullhypotheses}
\subsection{Motivation}
\label{sec:pi0motivation}
In this section we demonstrate how randomized $p$-values generally lead to a more precise estimation of the proportion $\pi_{0}$ of true null hypotheses. This is useful for data-adaptive multiple test procedures, but knowing $m_{0}=m \cdot \pi_{0}$ can also be valuable in itself. In bio-marker identification, for instance, the size of $m_1 = m - m_{0}$ can be an indicator for the complexity of the examined disease.

The Schweder-Spj{\o}tvoll estimator is given by $\hat{\pi}_{0} \equiv \hat{\pi}_{0}(\lambda)=\{1-\hat{F}_{m}(\lambda)\}/(1-\lambda)$, where $\hat{F}_{m}$ denotes the empirical cumulative distribution function of the $m$ marginal $p$-values, and $\lambda\in[0,1)$ is a tuning parameter \citep{schweder1982plots}. The estimator $\hat{\pi}_{0}(\lambda)$ represents the proportion of $p$-values above $\lambda$ divided by the expected proportion of the latter given uniformly distributed $p$-values. Assuming that the $p$-values corresponding to the false null hypotheses are always below $\lambda$, and the ones corresponding to the true null hypotheses are uniformly distributed on $[0,1]$, the term $1-\hat{F}_{m}(\lambda)$ is then, in expectation equal to $(1-\lambda)\pi_{0}$, leading to an unbiased estimator $\hat{\pi}_{0}(\lambda)$ for $\pi_{0}$. Graphically, the estimator $\hat{\pi}_{0}(\lambda)$ equals one minus the offset at $t=0$ of the straight line connecting $(\lambda,\hat{F}_{m}(\lambda))$ with $(1,1)$. We sometimes write $\hat{\pi}_{0}^{LFC}$ and $\hat{\pi}_{0}^{rand}$ to emphasize the usage of the LFC-based or the randomized $p$-values in the estimator $\hat{\pi}_{0}$, respectively.

\subsection{Simulations}
\label{sec:pi0simulation}

First, we simulated one realization of the empirical cumulative distribution functions of $(p_{j}^{LFC})_{j=1,\ldots,m}$ and $(p_{j}^{rand})_{j=1,\ldots,m}$, computed on the same data, where we chose $m=500$, $s=10$, $\gamma=6$, and $\pi_{0}=0.7$. Hence, we consider $10$ studies, each examining the same $500$ endpoints, where $m_{1}=150$ of these have a positive effect in at least $\gamma=6$ and the other $m_{0}=350$ have a positive effect in less than $6$ studies. We call these true and false endpoints, respectively, according to whether their respective null hypotheses are true or false. For each true and false endpoint we drew the number of studies with positive effects binomially from $\{0,\ldots,\gamma-1\}$ and $\{\gamma,\ldots,s\}$ with the success probabilities $p_{0}=0.8$ and $p_{1}=0.8$, respectively. For each study $i$ and endpoint $j$ we set the sample size to $n_{i,j}=50$ and drew for non-positive effects $\theta_{i,j}(\vartheta)$ uniformly from $(\mu_{\mathrm{min}}~50^{-1/2},0]$ and for positive effects uniformly from $(0,\mu_{\mathrm{max}}]$, where we chose $\mu_{\mathrm{min}}=-2.5$ and $\mu_{\mathrm{max}}=1.5$.

\begin{figure}
\includegraphics[width=250pt]{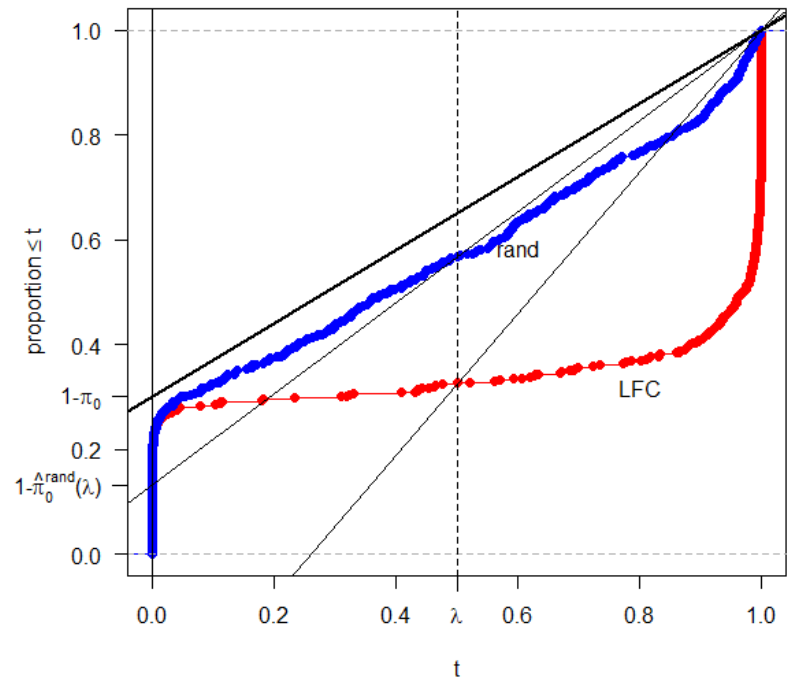}
\caption{One realization of the empirical cumulative distribution functions of the marginal $p$-values $(p_{j}^{LFC})_{j}$ and $(p_{j}^{rand})_{j}$, respectively, under Model $1$ for $m=500,s=10,\gamma=6$, and $\pi_{0}=0.7$. The thick, straight line connects the points $(1,1)$ and $(0,1-\pi_{0})$. The two thinner, straight lines connect the points $(1,1)$ and $(\lambda,\hat{F}_{m}(\lambda))$ and intersect the vertical axis at $(0,1-\hat{\pi}_{0}^{rand}(\lambda))$ or $(0,1-\hat{\pi}_{0}^{LFC}(\lambda))$ for the respective $p$-values.}
\label{fig:ecdfs}
\end{figure}

Figure~\ref{fig:ecdfs} displays one realization of the empirical cumulative distribution functions of the marginal, LFC-based and the marginal, randomized $p$-values, respectively. The estimation $\hat{\pi}_{0}(\lambda)$ is more accurate if the empirical cumulative distribution function of the utilized marginal $p$-values at point $t=\lambda$ is closer to the thick line connecting $(0,1-\pi_{0})$ with $(1,1)$. Clearly, $\hat{\pi}_{0}^{rand}(\lambda)$ is more accurate than $\hat{\pi}_{0}^{LFC}(\lambda)$ for $0.1<\lambda<1$. Also, $\hat{\pi}_{0}^{rand}(\lambda)$ is more stable with respect to $\lambda$, as the lower curvature of the respective empirical cumulative distribution function suggests.

Next, we calculated the expected values of $\hat{\pi}_{0}^{LFC}(\lambda)$ and $\hat{\pi}_{0}^{rand}(\lambda)$ for different values of $\pi_{0},(\mu_{\mathrm{min}},\mu_{\mathrm{max}}),$ and $\gamma$, where we set $s=10,m=100$, and $\lambda=1/2$. Apart from that, we drew everything else as before.

\begin{figure}
\includegraphics[width=300pt]{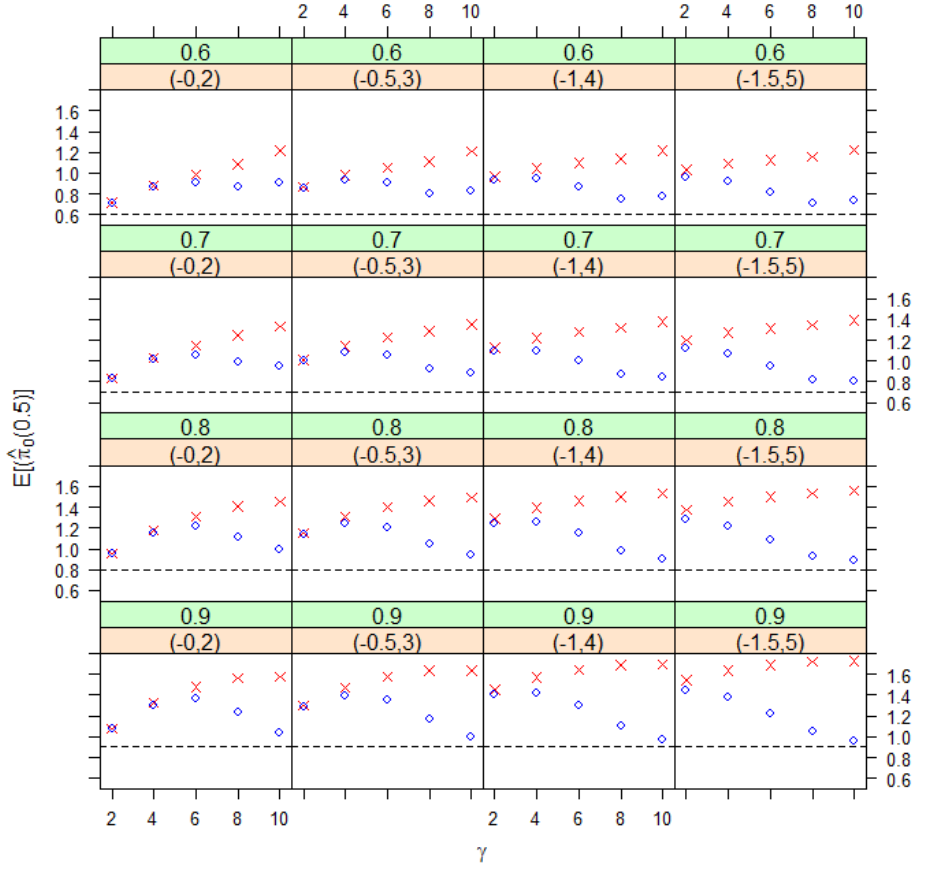}
\caption{A comparison of the expected values of $\hat{\pi}_{0}(1/2)$ utilizing either $(p_{j}^{LFC})_{j}$ (circles) or $(p_{j}^{rand})_{j}$ (crosses) in all considered settings. In each graph the horizontal axis displays the parameter $\gamma$. The graphs differ in their choice of $(\mu_{\mathrm{min}},\mu_{\mathrm{max}})$ (columns) and $\pi_{0}$ (rows). Dashed lines represent the true values of the proportion $\pi_{0}$ of true null hypotheses.}
\label{fig:epi0u}
\end{figure}

We looked at each combination of $\pi_{0}\in\{0.6,0.7,0.8,0.9\},\;(\mu_{\mathrm{min}},\mu_{\mathrm{max}})\in\{(0,2),(-0.5,3),(-1,4),(-1.5,5)\},\;\text{and}\;\gamma\in\{2,4,6,10\}$. Each pair $(\mu_{\mathrm{min}},\mu_{\mathrm{max}})$ was chosen such that $|\mu_{\mathrm{min}}|$ and $\mu_{\mathrm{max}}$ increase simultaneously, and thus, model uncertainty increases in both directions. Figure~\ref{fig:epi0u} illustrates the effect 
of $\gamma$ on the expected value of $\hat{\pi}_{0}(1/2)$ in each setting when utilizing LFC-based $p$-values (crosses) or randomized $p$-values (circles), respectively. For the exact numbers we refer to Table~\ref{tab:epi0lfc} and Table~\ref{tab:epi0rand}, respectively. All values have been double-checked by Monte Carlo simulations.

According to Lemma $1$ in \cite{RandoFisher}, the Schweder-Spjøtvoll estimator $\hat{\pi}_{0}(\lambda)$ applied to either of the $p$-values has a non-negative bias. In each setting we observe lower expected values and therefore lower bias for $\hat{\pi}_{0}^{rand}(1/2)$ than for $\hat{\pi}_{0}^{LFC}(1/2)$. The difference between the expectations tend to be more emphasized for higher $\gamma$ and higher model uncertainty, i.e. for larger $\mu_{\mathrm{max}}$ and $|\mu_{\mathrm{min}}|$.

As mentioned before, we expect a more stable estimation $\hat{\pi}_{0}(\lambda)$ of $\pi_{0}$ with respect to $\lambda$ when utilizing the randomized $p$-values. For the parameter settings $\pi_{0}=0.6,\gamma=8,\mu_{\mathrm{min}}=-2$, and $\mu_{\mathrm{max}}=4$, Fig.$~\ref{fig:lambda}$ compares the expected values of $\hat{\pi}_{0}(\lambda)$ for $\lambda=0.1,0.2,\ldots,0.9$ and either $p$-values. We checked many other configurations, too. They lead to similar results, although not always so pronounced.

\begin{figure}
\includegraphics[width=175pt]{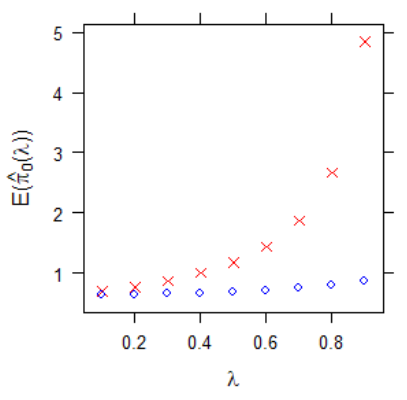}
\caption{The expected values of $\hat{\pi}_{0}(\lambda)$ for different tuning parameters $\lambda$ under Model $1$ for $m=100,s=10,n_{i,j}=50,\gamma=8,\pi_{0}=0.6,\mu_{\mathrm{min}}=-2,\mu_{\mathrm{max}}=4$, and $p_{0}=p_{1}=0.8$ when using either the LFC-based $p$-values (crosses) or the randomized $p$-values (circles).}
\label{fig:lambda}
\end{figure}

Finally, we examined the higher variance of $\hat{\pi}_{0}^{rand}(1/2)$ when utilizing the randomized $p$-values $(p_{j}^{rand})_{j}$, due to the additional randomization by $U_{j}\ (j=1,\ldots,m)$. We calculated the standard deviation of $\hat{\pi}_{0}(1/2)$ utilizing either the LFC-based $p$-values or the randomized $p$-values, for the same settings as we did for Fig.$~\ref{fig:epi0u}$ via Monte Carlo simulations. For the results we refer to the appendix. Using $(p_{j}^{rand})_{j=1,\ldots,m}$, we observe higher standard deviations of $\hat{\pi}_{0}^{rand}(1/2)$ in each setting short of one. The largest standard deviation when using the randomized $p$-values across all considered settings was below $0.1$. For the exact values we refer to \ref{sec:stochasticordering}. We also compared the mean squared errors of $\hat{\pi}_{0}^{rand}$ and $\hat{\pi}_{0}^{LFC}$ in all considered parameter settings. In each setting the mean squared error was higher when using the LFC-based $p$-values.

\section{An application on multiple Crohn's disease genome-wide assocation studies}
\label{sec:realdata}
We looked at the data from multiple genome-wide association studies with the goal of identifying susceptibility loci for Crohn's disease \citep{franke2010genome}. The authors looked at six distinct genome-wide association studies, further dividing two of these resulting in a total of eight distinct studies, which comprised $6,333$ disease cases and $15,056$ healthy controls altogether. In their discovery panel, they combined these eight studies in a meta-analysis and looked at the most promising features in a further replication panel. For lack of data on the latter part we only looked at the data stemming from the original eight studies.

In their work, the authors applied multiple $Z$-tests for the logarithmic odds ratios in each scan and combined them to test for two-sided associations of phenotype and genotype at each of $m$ loci. For these, randomized $p$-values can also be defined \citep{dickhaus2013randomized}. However, in such a two-sided setting each parameter in the null hypotheses $H_{j}=\{(\theta_{1,j},\ldots,\theta_{s,j})\in \mathbb{R}^{s}:\theta_{k,j}=0\;\text{for at least $s-\gamma+1$ indices $k$}\}\ (j=1,\ldots,m)$, would lie next to the respective alternative $K_{j}= \mathbb{R}^{s}\setminus H_{j}$ making each one an LFC for their respective null hypothesis. In spite of the composite nature of the null hypotheses, the LFC-based $p$-values would then hold a uniform distribution under any parameter in the null hypothesis and using randomized $p$-values would be unnecessary.

Instead, we looked at the original $Z$-scores for associations in one fixed direction between the investigated single-nucleotide polymorphisms and Crohn's disease. Each of the eight studies investigated the effect of $953,241$ single-nucleotide polymorphisms on Crohn's disease. We designated one of the studies as a primary study and selected the most promising features with the Benjamini-Hochberg step-up procedure at false discovery rate \citep{benjamini1995controlling} levels $q=0.2$ or $q=0.5$. After selection we ended up with $m=630$ and $m=2,257$ single-nucleotide polymorphisms, respectively, and tested their associations' replicability based on the remaining $s=7$ studies. For both false discovery rate levels $q$, we looked at the choices $\gamma=2$ and $\gamma=4$, and calculated the LFC-based and randomized $p$-values as in the model described in Section~\ref{sec:modelreplicability}. For these values of $\gamma$, we have $c_{j}=2^{-(7-2+1)}=2^{-6}$ and $c_{j}=2^{-(7-4+1)}=2^{-4}$, respectively, where $d_{j}=1/2$ results from the model $(j=1,\ldots,m)$.

We then calculated the Schweder-Spj{\o}tvoll estimator $\hat{\pi}_{0}(\lambda)$ with $\lambda=1/2$ for the four parameter settings. Figure~\ref{fig:realdata} illustrates the empirical cumulative distribution functions of the LFC-based and the randomized $p$-values, respectively, after selection. The values for the settings $(q,\gamma)=(0.2,2),(0.2,4),(0.5,2),(0.5,4)$ are, in order,
\begin{eqnarray*}
(\hat{\pi}_{0}^{LFC}(\lambda),E(\hat{\pi}_{0}^{rand}(\lambda))) &=& (0.4603,0.4651),(0.8857,0.7572),\\
&~&(0.9880,0.9668),(1.5498,1.3013),
\end{eqnarray*} 
where $E$ refers to the randomness of $(U_j: 1 \leq j \leq m)$.
These are also displayed above their corresponding graphs. The standard deviation for the estimation using the randomized $p$-values are $\text{var}^{1/2}(\hat{\pi}_{0}^{rand}(\lambda))=0.00276$, $0.01542$, $0.00377$, $0.01109$ for the respective settings in the same order. The values corresponding to the use of the randomized $p$-values are a result of Monte Carlo simulations with $100,000$ repetitions in each setting.

Let us discuss these results. An increase in the false discovery rate level $q$ increases the proportion $\pi_{0}$ which favours the use of the randomized $p$-values. A higher $\gamma$ increases the proportion $\pi_{0}$ and reduces the constant $c_{j}=2^{-(7-\gamma+1)}\ (j=1,\ldots,m)$, both benefiting the estimator $\hat{\pi}_{0}^{rand}(\lambda)$. Choosing $q$ and $\gamma$ both too high can lead to a too large $\pi_{0}$ making it difficult to estimate the latter as the example with $q=0.5$ and $\gamma=4$ demonstrates. On the other hand, choosing both $q$ and $\gamma$ too low results in a low proportion of true null hypotheses, of which the remaining do not offer high enough deviation from the alternative to facilitate the usage of randomized $p$-values as the example with $q=0.2$ and $\gamma=2$ demonstrates.

\begin{figure}
\includegraphics[width=325pt]{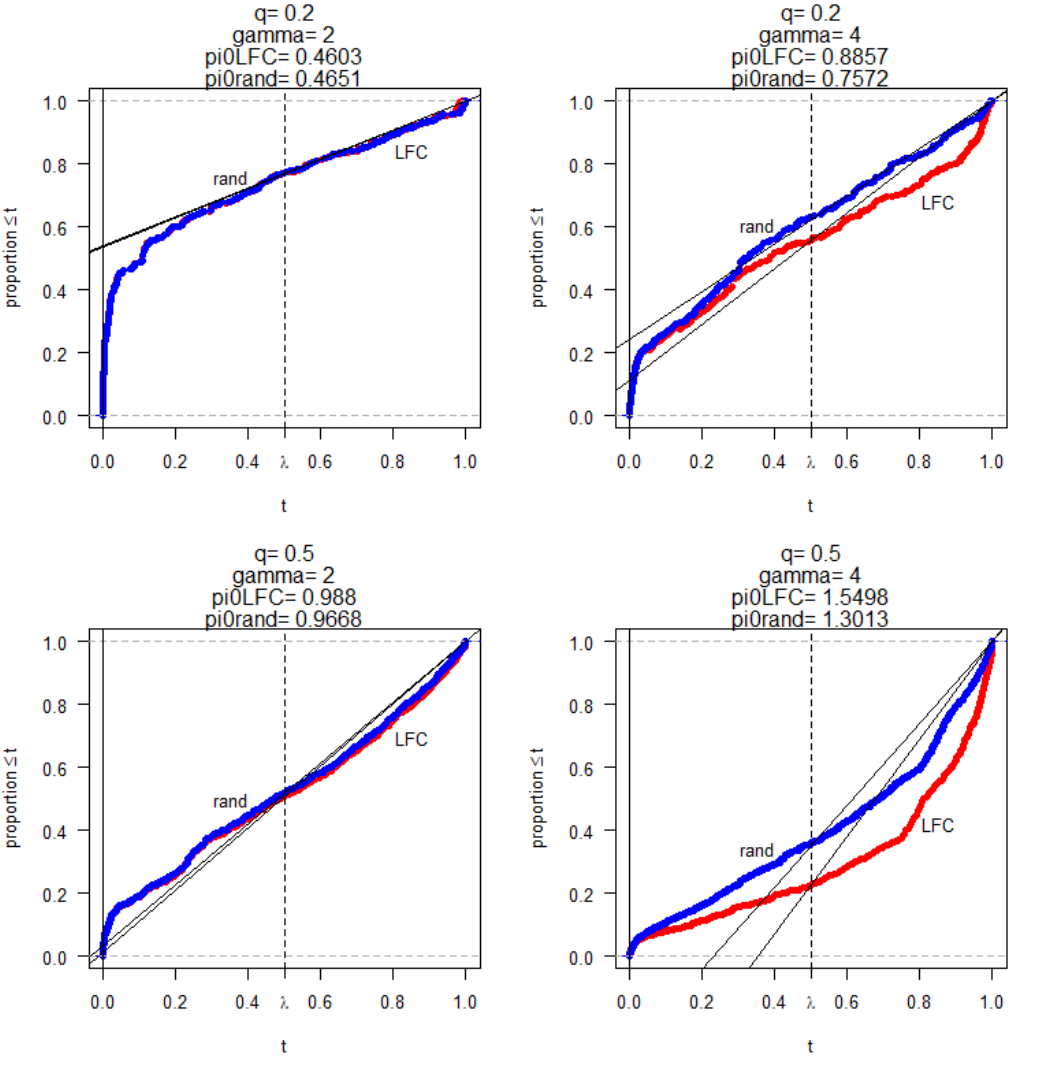}
\caption{The empirical cumulative distribution functions of the LFC-based and the randomized $p$-values, respectively, in the multiple Crohn's disease genome-wide associations studies example after selection. Selection has been conducted with the Benjamini--Hochberg step-up procedure with false discovery rates $q=0.2,0.5$, and the $p$-values are calculated according to the model as described in Section~\ref{sec:replicabilityanalysis} with $\gamma=2,4$. The straight lines connect the points $(1,1)$ and $(\lambda,\hat{F}_{m}(\lambda))$, and intersect the vertical axis in the point $(0,1-\hat{\pi}_{0}(\lambda))$, where $\lambda=1/2$. The values $\hat{\pi}_{0}^{LFC}(\lambda)$ and $E(\hat{\pi}_{0}^{rand}(\lambda))$ are displayed above their respective graphs as pi0LFC and pi0rand, respectively.}
\label{fig:realdata}
\end{figure}

\begin{table}[!htbp]
  \caption{Expected values of $\hat{\pi}_{0}^{LFC}(1/2)$ using the LFC-based $p$-values $(p_{j}^{LFC})_{j=1,\ldots,m}$ in Model $1$ with $s=10$}
  \centering
    \begin{tabular}{lrrrrr}
    $\gamma=2$   & \multicolumn{1}{l}{$\pi_{0}$} & 0.6   & 0.7   & 0.8   & 0.9 \\
    $(\mu_{\mathrm{min}},\mu_{\mathrm{max}})$ &       &       &       &       &  \\
    (0,2) &       & 0.71623177 & 0.83559977 & 0.95496776 & 1.07433576 \\
    (-0.5,3) &       & 0.86543138 & 1.00966879 & 1.1539062 & 1.29814362 \\
    (-1,4) &       & 0.96684663 & 1.12798719 & 1.28912775 & 1.4502683 \\
    (-1.5,5) &       & 1.02947056 & 1.20104866 & 1.37262676 & 1.54420486 \\
          &       &       &       &       &  \\
    $\gamma=4$   & \multicolumn{1}{l}{$\pi_{0}$} & 0.6   & 0.7   & 0.8   & 0.9 \\
    $(\mu_{\mathrm{min}},\mu_{\mathrm{max}})$ &       &       &       &       &  \\
    (0,2) &       & 0.88196066 & 1.02879483 & 1.17562899 & 1.32246315 \\
    (-0.5,3) &       & 0.98000457 & 1.14328717 & 1.30656978 & 1.46985239 \\
    (-1,4) &       & 1.04711715 & 1.22161107 & 1.39610499 & 1.57059891 \\
    (-1.5,5) &       & 1.08842236 & 1.26981052 & 1.45119869 & 1.63258685 \\
          &       &       &       &       &  \\
    $\gamma=6$   & \multicolumn{1}{l}{$\pi_{0}$} & 0.6   & 0.7   & 0.8   & 0.9 \\
    $(\mu_{\mathrm{min}},\mu_{\mathrm{max}})$ &       &       &       &       &  \\
    (0,2) &       & 0.98832493 & 1.15060573 & 1.31288654 & 1.47516734 \\
    (-0.5,3) &       & 1.05095857 & 1.22514476 & 1.39933094 & 1.57351713 \\
    (-1,4) &       & 1.09638526 & 1.27857305 & 1.46076084 & 1.64294864 \\
    (-1.5,5) &       & 1.12467501 & 1.31176513 & 1.49885525 & 1.68594537 \\
          &       &       &       &       &  \\
    $\gamma=8$   & \multicolumn{1}{l}{$\pi_{0}$} & 0.6   & 0.7   & 0.8   & 0.9 \\
    $(\mu_{\mathrm{min}},\mu_{\mathrm{max}})$ &       &       &       &       &  \\
    (0,2) &       & 1.08287299 & 1.24400739 & 1.40514179 & 1.56627619 \\
    (-0.5,3) &       & 1.10933717 & 1.28409891 & 1.45886065 & 1.63362239 \\
    (-1,4) &       & 1.13659204 & 1.31938899 & 1.50218594 & 1.68498288 \\
    (-1.5,5) &       & 1.15472527 & 1.34232416 & 1.52992304 & 1.71752192 \\
          &       &       &       &       &  \\
    $\gamma=10$  & \multicolumn{1}{l}{$\pi_{0}$} & 0.6   & 0.7   & 0.8   & 0.9 \\
    $(\mu_{\mathrm{min}},\mu_{\mathrm{max}})$ &       &       &       &       &  \\
    (0,2) &       & 1.21412161 & 1.33352516 & 1.45292871 & 1.57233226 \\
    (-0.5,3) &       & 1.20689424 & 1.35040322 & 1.4939122 & 1.63742118 \\
    (-1,4) &       & 1.21648791 & 1.37509473 & 1.53370156 & 1.69230838 \\
    (-1.5,5) &       & 1.22408055 & 1.39233485 & 1.56058915 & 1.72884345 \\
    \end{tabular}
  \label{tab:epi0lfc}
\end{table}%

\newpage

\begin{table}[!htbp]
\caption{Expected values of $\hat{\pi}_{0}^{rand}(1/2)$ using $(p_{j}^{rand})_{j=1,\ldots,m}$ in Model $1$ with $s=10$}
  \centering
    \begin{tabular}{lrrrrr}
    $\gamma=2$   & \multicolumn{1}{l}{$\pi_{0}$} & 0.6   & 0.7   & 0.8   & 0.9 \\
    $(\mu_{\mathrm{min}},\mu_{\mathrm{max}})$ &       &       &       &       &  \\
    (0,2) &       & 0.7148753 & 0.8340172 & 0.9531591 & 1.072301 \\
    (-0.5,3) &       & 0.85571927 & 0.998338 & 1.14095672 & 1.28357545 \\
    (-1,4) &       & 0.93561312 & 1.09154809 & 1.24748306 & 1.40341803 \\
    (-1.5,5) &       & 0.9608599 & 1.1210029 & 1.28114589 & 1.44128888 \\
          &       &       &       &       &  \\
    $\gamma=4$   & \multicolumn{1}{l}{$\pi_{0}$} & 0.6   & 0.7   & 0.8   & 0.9 \\
    $(\mu_{\mathrm{min}},\mu_{\mathrm{max}})$ &       &       &       &       &  \\
    (0,2) &       & 0.86597221 & 1.01013914 & 1.15430607 & 1.298473 \\
    (-0.5,3) &       & 0.93117086 & 1.08631394 & 1.24145702 & 1.3966001 \\
    (-1,4) &       & 0.94324531 & 1.10042711 & 1.25760891 & 1.41479072 \\
    (-1.5,5) &       & 0.91686267 & 1.06965757 & 1.22245248 & 1.37524738 \\
          &       &       &       &       &  \\
    $\gamma=6$   & \multicolumn{1}{l}{$\pi_{0}$} & 0.6   & 0.7   & 0.8   & 0.9 \\
    $(\mu_{\mathrm{min}},\mu_{\mathrm{max}})$ &       &       &       &       &  \\
    (0,2) &       & 0.91209707 & 1.06154477 & 1.21099247 & 1.36044016 \\
    (-0.5,3) &       & 0.90327673 & 1.05281018 & 1.20234362 & 1.35187706 \\
    (-1,4) &       & 0.8645 & 1.00802616 & 1.15155232 & 1.29507848 \\
    (-1.5,5) &       & 0.81444879 & 0.94983046 & 1.08521213 & 1.2205938 \\
          &       &       &       &       &  \\
    $\gamma=8$   & \multicolumn{1}{l}{$\pi_{0}$} & 0.6   & 0.7   & 0.8   & 0.9 \\
    $(\mu_{\mathrm{min}},\mu_{\mathrm{max}})$ &       &       &       &       &  \\
    (0,2) &       & 0.86888431 & 0.99140844 & 1.11393257 & 1.2364567 \\
    (-0.5,3) &       & 0.80273805 & 0.92509235 & 1.04744665 & 1.16980095 \\
    (-1,4) &       & 0.74938827 & 0.86699005 & 0.98459182 & 1.10219359 \\
    (-1.5,5) &       & 0.70806821 & 0.82087336 & 0.9336785 & 1.04648365 \\
          &       &       &       &       &  \\
    $\gamma=10$  & \multicolumn{1}{l}{$\pi_{0}$} & 0.6   & 0.7   & 0.8   & 0.9 \\
    $(\mu_{\mathrm{min}},\mu_{\mathrm{max}})$ &       &       &       &       &  \\
    (0,2) &       & 0.91207747 & 0.95515966 & 0.99824186 & 1.04132406 \\
    (-0.5,3) &       & 0.8274848 & 0.88521938 & 0.94295397 & 1.00068855 \\
    (-1,4) &       & 0.77492519 & 0.84086328 & 0.90680137 & 0.97273946 \\
    (-1.5,5) &       & 0.74041001 & 0.81169235 & 0.88297469 & 0.95425703 \\
    \end{tabular}%
    \label{tab:epi0rand}%
\end{table}%

\section{Discussion}
\label{sec:discussion}

In the context of simultaneous testing of composite null hypotheses, we have demonstrated that the usage of randomized $p$-values leads to a more accurate estimation of $\pi_{0}$ when compared with the usage of LFC-based $p$-values. We have explicitly demonstrated this for the Schweder-Spj{\o}tvoll estimator $\hat{\pi}_{0}$. The higher estimation variances induced by the uniform random variates used for randomization are in most cases negligible, so that the mean squared error is lower for $\hat{\pi}_{0}^{rand}$ than for $\hat{\pi}_{0}^{LFC}$.

Our theory applies to any choice of the parameter $\gamma=2,\ldots,s$. We have not further discussed the choice of $\gamma$ nor do we make recommendations in this work. Choosing $\gamma$ close to $s$ results in strong replicability statements, but potentially only few rejections. On the other hand, in the presence of a very large number of studies $s$, replicability statements may not be suitable when choosing $\gamma=2$. Thus, one could make $\gamma$ dependent on $s$, like $\gamma=\beta s$ for $\beta\in(0,1)$. Alternatively, instead of pre-defining $\gamma$, we could for each $j=1,\ldots,m$ determine the largest $\gamma=\gamma(j)$, for which we would still reject $H_{j}$. It is then possible to declare replicability for endpoint $j$ if $\gamma(j)/s>\beta$ holds, where $\beta\in(0,1)$ is pre-defined. 

Furthermore, we have not discussed the incorporation of the estimated proportion of true null hypotheses in so-called adaptive multiple tests. \cite{blanchard2009adaptive} presented a categorization of adaptive procedures that divide between plug-in, two-stage and one-stage procedures, and provided adaptive procedures that control the false discovery rate. \cite{finner2009controlling} investigated the problem of controlling the family-wise error rate when using an estimator of $\pi_{0}$ as a plug-in estimator in single-step or step-down procedures. \cite{bogomolov2018assessing} gave an adaptive procedure that incorporates estimations of the proportion of true null hypotheses among the selected features and controls the false discovery rate for replicability analysis with two studies. It remains to be investigated to what extent the usage of randomized $p$-values can improve the power of such adaptive procedures. In the case of $s = 1$, some results in this direction can be found in \cite{dickhaus2013randomized}. These results indicate, that the power gain can be substantial.

Finally, one challenging extension of our proposed methodology is to investigate randomized $p$-values for other types of summary statistics, in particular combination test statistics of Fisher- or Stouffer-Liptak-type; see, e.\ g., \cite{Zwet1967}, \cite{Kim-Stouffer} and the references therein. In \ref{sup:sec:furthersimulations} we compare their (non-randomized) use in $\hat{\pi}_{0}$ with the use of our proposed randomized $p$-values that result from our summary statistics. Under the same model and considering the same parameter settings as in Section~\ref{sec:pi0simulation} the use of the randomized $p$-values in the Schweder-Spj{\o}tvoll estimator is still more accurate in most cases. Another possibility in this direction is to consider statistics derived from Bayesian models, for instance local false discovery rates or Bayes factors, as in \cite{Yekutieli-TEST} and \cite{Dickhaus-Bayes}, respectively.

\section*{Acknowledgments}
Financial support by the German Research Foundation under grant No. DI 1723/5-1 is gratefully acknowledged.

\appendix


\section{Some concepts of stochastic ordering} \label{sec:stochasticordering}

We briefly introduce some concepts of stochastic ordering and notations. For some further results we refer to \ref{sup:sec:resultsstochasticorders}.
\begin{definition} 
\label{def:stochasticorders}\ \\
Let $X,Y$ be two random variables with cumulative distribution functions $F,G$, respectively.
\begin{description}
\item[$1.$] We say \textit{$X$ is smaller than $Y$ in the usual stochastic order} or $X$ is stochastically not larger than $Y$, denoted by $X\leq_{\mathrm{st}}Y$, if and only if it holds $F(x)\geq G(x)$ for all $x\in(-\infty,\infty)$.

Intuitively, $X$ is more likely than $Y$ to take on small values.
\item[$2.$] We say \textit{$X$ is smaller than $Y$ in the hazard rate order}, denoted by $X\leq_{\mathrm{hr}}Y$, if and only if $(1-G(t))/(1-F(t))$ does not decrease in $t<\mathrm{max}\{u(X),u(Y)\}$, where $u(X),u(Y)$ denote the right endpoints of the supports of $X,Y$, respectively. We define $a/0=\infty$, whenever $a>0$.

Equivalently, if $X$ and $Y$ admit Lebesgue-density functions $f,g$, respectively, it holds $X\leq_{\mathrm{hr}}Y$ if and only if $f(t)/(1-F(t))\geq g(t)/(1-G(t))$ for all $t\in \mathbb{R}$, i.e. $Y$ has a smaller hazard rate function.
\item[$3.$] If $X,Y$ admit Lebesgue-density functions $f,g$, respectively, we say \textit{$X$ is smaller than $Y$ in the likelihood ratio order}, denoted by $X\leq_{\mathrm{lr}}Y$, if and only if $g(t)/f(t)$ is non-decreasing in $t$ over the union of the supports of $X$ and $Y$, where $a/0=\infty$, whenever $a>0$. Equivalently, it holds $X\leq_{lr}Y$ if and only if $f(y)g(x)\leq f(x)g(y),\;\text{for all}\;x\leq y$ 
\end{description} 
\end{definition}
These three orders only depend on the distributions of $X,Y$, i.e. they only depend on $F,G,f,g$. Hence, we introduce the following notations.
\begin{definition}
\label{def:stochasticordersnotation}

Given a statistical model $\big(\Omega,\mathcal{F},(\text{pr}_{\vartheta})_{\vartheta\in\Theta}\big)$ and test statistics $T,S:\Omega\to \mathbb{R}$ with cumulative distribution functions $F_{\vartheta},G_{\vartheta}$, respectively, and Lebesgue-density functions $f_{\vartheta},g_{\vartheta}$, respectively, under $\vartheta\in\Theta$, we write $T^{(\vartheta_{1})}\leq_{\mathrm{st}}S^{(\vartheta_{2})}$, if it holds $F_{\vartheta_{1}}(x)\geq G_{\vartheta_{2}}(x)$, for all $x$ and parameters $\vartheta_{1},\vartheta_{2}\in\Theta$. Analogously, we denote $T^{(\vartheta_{1})}\leq_{\mathrm{hr}}S^{(\vartheta_{2})}$, or $T^{(\vartheta_{1})}\leq_{\mathrm{lr}}S^{(\vartheta_{2})}$, if $F_{\vartheta_{1}},G_{\vartheta_{2}},f_{\vartheta_{1}},g_{\vartheta_{2}}$ satisfy the corresponding requirements for parameters $\vartheta_{1},\vartheta_{2}\in\Theta$.
\end{definition}

\section{Some results regarding stochastic orders} \label{sup:sec:resultsstochasticorders}

We introduce some results regarding the hazard rate order. For a set of random variables $Z_{1},\ldots,Z_{n}$, $n\geq 2$, we denote the order statistics of the first $m\leq n$ $Z_{i}$'s by $Z_{(1:m)}\leq\cdots\leq Z_{(m:m)}$. For $m=n$ we usually write $Z_{(1)}\leq\cdots\leq Z_{(n)}$.
\begin{theorem} \label{thm:hazardorder}\ \\
Let $X_{1},\ldots,X_{n}$ and $Y_{1},\ldots,Y_{n}$, be two sets of independent, not necessarily identically distributed, random variables.
\begin{description}
\item[$1$.] \citep[Theorem 1.B.28]{shaked2007stochastic}\\ It holds $X_{(k:m)}\leq_{\mathrm{hr}}X_{(k:m-1)}\ (k=1,\ldots,m-1)$.
\item[$2.$] \citep[Theorem 1.B.35]{shaked2007stochastic}\\ If $X_{1},\ldots,X_{n},Y_{1},\ldots,Y_{n}$ all have the same support $(a,b)$ for some $a<b$, and $X_{i}\leq_{\mathrm{hr}}Y_{j}\ (i=1,\ldots,n;\ j=1,\ldots,n)$, then $X_{(k:n)}\leq_{\mathrm{hr}}Y_{(k:n)}\ (k=1,\ldots,n)$. 
\item[$3.$] \citep[Theorem 1.B.2]{shaked2007stochastic}\\
If $X\leq_{\mathrm{hr}}Y$ and $\psi$ is an increasing function, then $\psi(X)\leq_{\mathrm{hr}}\psi(Y)$.
\end{description}
\end{theorem}
For proofs and further details, the reader may consult Chapter $1.B.$ and Chapter $1.C.$ in \cite{shaked2007stochastic}.

Now, let $X_{1},\ldots,X_{n}$ be independent random variables with support $(0,1)$ and $U_{1},\ldots,U_{n}$ be independent, uniformly distributed random variables on $[0,1]$.
\begin{lemma}
\label{lm:hazard}\ \\
For all fixed $n\geq 2$, $i\in\{1,\ldots,n\}$, if $X_{k}\leq_{\mathrm{hr}}U_{k}$ holds for at least $i$ indices $k\in\{1,\ldots,n\}$, then $X_{(i:n)}\leq_{\mathrm{hr}}U_{(i:i)}$.
\end{lemma}
\begin{proof}
At first we consider the case $i=n$, that is, we assume $X_{k}\leq_{\mathrm{hr}}U_{k}$ holds for all $k=1,\ldots,n$. Then, we have $X_{i}\leq_{\mathrm{hr}}U_{j}$ for all $i,j$, since the hazard rate order only depends on the distributions of $X_{i}$ and $U_{j}$, and therefore $X_{(n:n)}\leq_{\mathrm{hr}}U_{(n:n)}$ follows directly from Part $2$ of Theorem~\ref{thm:hazardorder}.

For $i=1,\ldots,n-1$, we obtain from Part $1$ of Theorem~\ref{thm:hazardorder}, that $X_{(i:n)}\leq_{\mathrm{hr}}X_{(i:n-1)}\leq_{\mathrm{hr}}\cdots\leq_{\mathrm{hr}}X_{(i:i)}\leq_{\mathrm{hr}}U_{(i:i)}$, where the last inequality follows from the first part if $X_{k}\leq_{\mathrm{hr}}U_{k}$ holds for $k=1,\ldots,i$. Since $X_{1},\ldots,X_{n}$ were assumed to have $i$ such $X_{k}$, and prior calculations hold for any order of $X_{1},\ldots,X_{n}$, we can assume $X_{k}\leq_{\mathrm{hr}}U_{k}\ (k=1,\ldots,i)$, as desired.
\end{proof}
This lemma can be extended to any stochastically independent and identically distributed $U_{1},\ldots,U_{k}$ with support $(0,1)$ or any support $(a,b)$ shared with $X_{1},\ldots,X_{n}$.

The following theorem is due to \cite[Theorem 1.C.2]{shaked2007stochastic} and establishes a relationship between the three stochastic orders presented in Definition~\ref{def:stochasticorders}.
\begin{theorem*}
\label{thm:orderoforders}
For two continuous random variables $X,Y$ the likelihood ratio order $X\leq_{\mathrm{lr}}Y$ implies the hazard rate order $X\leq_{\mathrm{hr}}Y$. Both imply the stochastic order $X\leq_{\mathrm{st}}Y$.
\end{theorem*}

\section{Proofs} \label{sup:sec:proofs}

\subsection*{Proof of Theorem~\ref{thm:1}}
\label{app:proofthm1}
In order to show the first assertion, we notice that, due to assumption $(GA1)$, it holds $\{x\in\Omega:\;T_{j}(x)\in\Gamma_{j}(c_{j})\}=\{x\in\Omega:\;\hat{\theta}_{j}(x)\in K_{j}\}$. This implies
\begin{equation*}
\text{pr}_{\vartheta_{0}}(\hat{\theta}_{j}(X)\in K_{j})=\text{pr}_{\vartheta_{0}}(T_{j}(X)\in\Gamma_{j}(c_{j}))=\underset{\vartheta:\theta_{j}(\vartheta)\in H_{j}}{\mathrm{sup}}\text{pr}_{\vartheta}(T_{j}(X)\in\Gamma_{j}(c_{j}))=c_{j}.
\end{equation*}

Regarding the second assertion, we obtain that
\begin{equation}
G_{j}(t)=\text{pr}_{\vartheta_{0}}(p_{j}^{LFC}(X)\leq t\mid\hat{\theta}_{j}(X)\in K_{j})=\frac{\text{pr}_{\vartheta_{0}}(p_{j}^{LFC}(X)\leq t,\;\hat{\theta}_{j}(X)\in K_{j})}{\text{pr}_{\vartheta_{0}}(\hat{\theta}_{j}(X)\in K_{j})}. \label{eq:Gj}
\end{equation}
From $(GA1)$ it is $\{\hat{\theta}_{j}(X)\in K_{j}\}=\{T_{j}(X)\in\Gamma_{j}(c_{j})\}$. With that in mind, it is easy to see that it holds
\begin{align}
\left\{
\begin{array}{l l}
\hat{\theta}_{j}(x)\in K_{j}\Longrightarrow p_{j}^{LFC}(x)\leq c_{j}\Longrightarrow p_{j}^{LFC}(x)\leq t,\quad &t\geq c_{j},\\
p_{j}^{LFC}(x)\leq t\Longrightarrow p_{j}^{LFC}(x)< c_{j}\Longrightarrow\hat{\theta}_{j}(x)\in K_{j},\quad &t<c_{j},
\end{array}
\right. \label{eq:thm1eq1}
\end{align}
for all $x\in\Omega$. Consequently, the numerator on the right hand side in \eqref{eq:Gj} is either $\text{pr}_{\vartheta_{0}}(p_{j}^{LFC}(X)\leq t)=t$ or $\text{pr}_{\vartheta_{0}}(\hat{\theta}_{j}\in K_{j})$ for $t<c_{j}$ and $t\geq c_{j}$, respectively. This  leads to
\begin{align*}
G_{j}(t)&=\frac{t}{\text{pr}_{\vartheta_{0}}(\hat{\theta}_{j}\in K_{j})}\textbf{1}_{[0,c_{j})}(t)+\textbf{1}_{[c_{j},1]}(t)\\
&=\frac{t}{\text{pr}_{\vartheta_{0}}(\hat{\theta}_{j}\in K_{j})}\textbf{1}_{[0,\text{pr}_{\vartheta_{0}}(\hat{\theta}_{j}\in K_{j}))}(t)+\textbf{1}_{[\text{pr}_{\vartheta_{0}}(\hat{\theta}_{j}\in K_{j}),1]}(t).
\end{align*}

Finally, we show the third assertion.  Using Part $2$, we only have to show, that $\hat{\theta}_{j}(x)\in K_{j}$ implies $p_{j}^{LFC}(x)\leq c_{j}$ for all $x\in\Omega$, which is already part of \eqref{eq:thm1eq1}.

\subsection*{Proof of Theorem~\ref{thm:2}}
\label{app:proofthm2}
We recall from Theorem~\ref{thm:1} that
\begin{equation*}
p_{j}^{rand}(X,U_{j})=U_{j}~\textbf{1}_{H_{j}}\{\hat{\theta}_{j}(X)\}+\frac{p_{j}^{LFC}(X)}{\text{pr}_{\vartheta_{0}}(\hat{\theta}_{j}\in K_{j})}~\textbf{1}_{K_{j}}\{\hat{\theta}_{j}(X)\},
\end{equation*}
which implies 
\begin{align}
\label{eq:thm2introduction}
\text{pr}_{\vartheta}\big(p_{j}^{rand}(X,U_{j})\leq t\big)&=t~\text{pr}_{\vartheta}\big(\hat{\theta}_{j}(X)\in H_{j}\big)\\
&+\text{pr}_{\vartheta}\bigg[\frac{p_{j}^{LFC}(X)}{c_{j}}~\textbf{1}_{K_{j}}\{\hat{\theta}_{j}(X)\}\leq t\bigg],\quad t\in[0,1].\nonumber
\end{align}
Now, $\text{pr}_{\vartheta}\big(p_{j}^{rand}(X,U_{j})\leq t\big)\leq t$ from \eqref{eq:thm2introduction} holds, if and only if for the second summand in $(\ref{eq:thm2introduction})$ 
\begin{equation}
\label{eq:thm2equialenttovalidity}
\text{pr}_{\vartheta}\bigg[\frac{p_{j}^{LFC}(X)}{c_{j}}\textbf{1}_{K_{j}}\{\hat{\theta}_{j}(X)\}\leq t\bigg]\leq t~\text{pr}_{\vartheta}\big(\hat{\theta}_{j}(X)\in K_{j}\big)
\end{equation}
is fulfilled. Note, that due to assumption $(GA1)$ the term $\textbf{1}_{K_{j}}\{\hat{\theta}_{j}(X)\}$ in \eqref{eq:thm2equialenttovalidity} can be omitted.

The statement in Theorem~\ref{thm:2} was that
\begin{equation*}
\text{pr}_{\vartheta}\big(T_{j}(X)\in\Gamma_{j}(z)\big)\leq z\,\frac{\text{pr}_{\vartheta}(\hat{\theta}_{j}\in K_{j})}{\text{pr}_{\vartheta_{0}}(\hat{\theta}_{j}\in K_{j})},\quad 0\leq z\leq \text{pr}_{\vartheta_{0}}(\hat{\theta}_{j}\in K_{j}),
\end{equation*}
is equivalent to the validity of $p_{j}^{rand}$.

This follows from \eqref{eq:thm2equialenttovalidity} when substituting $z=t\,c_{j}=t\,\text{pr}_{\vartheta_{0}}(\hat{\theta}_{j}\in K_{j})$ and by seeing that $\text{pr}_{\vartheta}\big(p_{j}^{LFC}(X)\leq t\big)=\text{pr}_{\vartheta}\big(T_{j}(X)\in\Gamma_{j}(t)\big),\;t\in[0,1]$, holds. 
\subsection*{Proof of Theorem~\ref{thm:3}}
\label{app:proofthm3}
At first we show that 
\begin{equation}
\label{eq:thm3eq3}
\text{pr}_{\vartheta}\big(T_{j}(X)>z\big)\leq \text{pr}_{\vartheta_{0}}\big(T_{j}(X)>z\big)~\frac{\text{pr}_{\vartheta}(\hat{\theta}_{j}\in K_{j})}{\text{pr}_{\vartheta_{0}}(\hat{\theta}_{j}\in K_{j})},\;z\in[F^{-1}(1-c_{j}),\infty],
\end{equation}
holding for any $\vartheta\in\Theta$ with $\theta_{j}(\vartheta)\in H_{j}$ is equivalent to the validity of $p_{j}^{rand}$.

We make use of the following auxiliary result.
\begin{lemma}
\label{lm:thm2lemma1}
Let $h_{\vartheta}:[0,1]\to[0,1]$ be defined as follows
\begin{equation*}
h_{\vartheta}(z)=\text{pr}_{\vartheta}\Big(T_{j}(X)\in\Gamma_{j}\big(z\ \text{pr}_{\vartheta_{0}}(\hat{\theta}_{j}\in K_{j})\big)\Big)-z\ \text{pr}_{\vartheta}(\hat{\theta}_{j}\in K_{j}).
\end{equation*} 
Then, for all $\vartheta\in\Theta$ with $\theta_{j}(\vartheta)\in H_{j}$, it holds $h_{\vartheta}(0)=h_{\vartheta}(1)=0$.
\end{lemma}
\begin{proof}
We see that $h_{\vartheta}(0)=\text{pr}_{\vartheta}\big(T_{j}(X)\in\Gamma_{j}(0)\big)=0$. Due to assumption $(GA1)$ and Theorem~\ref{thm:1} it holds
\begin{equation*}
\big\{x\in\Omega:T_{j}(x)\in\Gamma_{j}\big(\underbrace{\text{pr}_{\vartheta_{0}}(\hat{\theta}_{j}\in K_{j})}_{c_{j}}\big)\big\}=\big\{x\in\Omega:\hat{\theta}_{j}(x)\in K_{j}\big\},
\end{equation*}
which implies $h_{\vartheta}(1)=0$.
\end{proof}
The condition $h_{\vartheta}\leq 0$ for all $\vartheta$ with $\theta_{j}(\vartheta)\in H_{j}$, is equivalent to the condition in Theorem~\ref{thm:2}, and hence equivalent to the validity of $p_{j}^{rand}$.

With our condition to the rejection regions $\Gamma_{j}$, it holds
\begin{align}
h_{\vartheta}(t)&=\text{pr}_{\vartheta}\Big[T_{j}(X)\in\Gamma_{j}\big(t~ \text{pr}_{\vartheta_{0}}(\hat{\theta}_{j}\in K_{j})\big)\Big]-t~\text{pr}_{\vartheta}(\hat{\theta}_{j}\in K_{j})\nonumber \\ 
&=\text{pr}_{\vartheta}\Big[T_{j}(X)>F^{-1}\big(1-t~\text{pr}_{\vartheta_{0}}(\hat{\theta}_{j}\in K_{j})\big)\Big]-t~\text{pr}_{\vartheta}(\hat{\theta}_{j}\in K_{j})\label{eq:thm3eq1}.
\end{align}
Substituting $z=F^{-1}\big(1-t~\text{pr}_{\vartheta_{0}}(\hat{\theta}_{i}\in K_{i})\big)$ in 
\eqref{eq:thm3eq1}, we obtain that
\begin{align}
h_{\vartheta}(t)&=\text{pr}_{\vartheta}(T_{j}(X)>z)-(1-F(z))\frac{\text{pr}_{\vartheta}(\hat{\theta}_{j}\in K_{j})}{\text{pr}_{\vartheta_{0}}(\hat{\theta}_{j}\in K_{j})}\nonumber\\
&=\text{pr}_{\vartheta}(T_{j}(X)>z)-\text{pr}_{\vartheta_{0}}(T_{j}(X)>z)\frac{\text{pr}_{\vartheta}(\hat{\theta}_{j}\in K_{j})}{\text{pr}_{\vartheta_{0}}(\hat{\theta}_{j}\in K_{j})}\label{eq:thm3eq2},
\end{align}
and thus $h_{\vartheta}(t)\leq 0$ for all $t\in[0,1]$ if and only if \eqref{eq:thm3eq3} holds. Furthermore, from assumption $(GA1)$ it holds $\big\{\hat{\theta}_{j}\in K_{j}\big\}=\big\{T_{j}(X)\in\Gamma_{j}(c_{j})\big\}=\big\{T_{j}(X)>F^{-1}(1-c_{j})=:a\big\}$, which implies, that \eqref{eq:thm3eq3} is equivalent to 
\begin{equation}
\label{eq:thm3eq4}
\frac{\text{pr}_{\vartheta}\big(T_{j}(X)>a+b\big)}{\text{pr}_{\vartheta}\big(T_{j}(X)>a\big)}\leq\frac{\text{pr}_{\vartheta_{0}}\big(T_{j}(X)>a+b\big)}{\text{pr}_{\vartheta_{0}}\big(T_{j}(X)>a\big)},\;\text{for all}\;b>0.
\end{equation}
Now $T_{j}(X)^{(\vartheta)}\leq_{\mathrm{hr}}T_{j}(X)^{(\vartheta_{0})}$ is equivalent to \eqref{eq:thm3eq4} holding for any $a$, and thus, it implies \eqref{eq:thm3eq3} and therefore the validity of $p_{j}^{rand}$.

\subsection*{Proof of Theorem~\ref{thm:stochasticorder}}
\label{app:proofthmstochasticorder}
Let a model as in Section$~\ref{sec:modelsetup}$ be given and $j\in\{1,\ldots,m\}$ be fixed. It is $p_{j}^{rand}(X,U_{j})=U_{j}\textbf{1}\{p_{j}^{LFC}>c_{j}\}+p_{j}^{LFC}c_{j}^{-1}\textbf{1}\{p_{j}^{LFC}(X)\leq c_{j}\}$ almost surely. We introduce the notation $p(X,U_{j},c)=U_{j}\textbf{1}\{p_{j}^{LFC}>c\}+p_{j}^{LFC}c^{-1}\textbf{1}\{p_{j}^{LFC}(X)\leq c\}$ for any $c\in[0,1]$. Note, that $p(X,U_{j},0)=U_{j}$ and $p(X,U_{j},1)=p_{j}^{LFC}(X)$.

For given $t\in[0,1]$ and $\vartheta\in\Theta$ we look at the function $c\mapsto h(c)=\mathbb{P}_{\vartheta}(p(X,U_{j},c)\leq t)$. We want to show that $h$ is non-decreasing if the cumulative distribution function of $p_{j}^{LFC}$ is convex and non-increasing if it is concave under $\vartheta$. It holds 
\begin{equation*}
h(c)=t\,\mathbb{P}_{\vartheta}(p_{j}^{LFC}(X)>c)+\mathbb{P}_{\vartheta}(p_{j}^{LFC}(X)\leq c\,t)
\end{equation*}
and 
\begin{equation*}
h'(c)=-t\,f_{\vartheta}(c)+t\,f_{\vartheta}(c\,t),
\end{equation*}
where $f_{\vartheta}$ is the density of $p_{j}^{LFC}(X)$ under $\vartheta$.

Now, if the cumulative distribution function of $p_{j}^{LFC}$ is convex under $\vartheta$, then $f_{\vartheta}$ is a non-decreasing function and $f_{\vartheta}(c\,t)\leq f_{\vartheta}(c)$ for all $c$, and analogously $f_{\vartheta}(c\,t)\geq f_{\vartheta}(c)$ for all $c$, if the cumulative distribution function of $p_{j}^{LFC}$ is concave under $\vartheta$.

\subsection*{Proof of Lemma~\ref{lm:replicabilitygeneralassumptions}}
\label{app:prooflm3}
We start with assumption $(GA1)$. It holds $\hat{\theta}_{j}(x)\in K_{j}$ if and only if $\hat{\theta}_{i,j}(x)>0$ for at least $\gamma$ indices $i\in\{1,\ldots,s\}$. Due to assumption $(RA2)$, the latter holds if and only if $p_{i,j}(x)<d_{j}$ for at least $\gamma$ indices $i\in\{1,\ldots,s\}$, which is equivalent to $1-p_{(\gamma),j}(x)>1-d_{j}$. Furthermore, $T_{j}(x)\in\Gamma_{j}(\alpha)$ is equivalent to $1-p_{(\gamma),j}(x)>F_{\mathrm{Beta}(s-\gamma+1,1)}^{-1}(1-\alpha)$, such that for $c_{j}=1-F_{\mathrm{Beta}(s-\gamma+1,1)}(1-d_{j})$, assumption $(GA1)$ is satisfied, i.e. $\{x\in\Omega:\;T_{j}(x)\in\Gamma_{j}(c_{j})\}=\{x\in\Omega:\;\hat{\theta}_{j}(x)\in K_{j}\}$.

For the verification of $(GA2)$ (nested rejection regions), we see that for every $j\in\{1,\ldots,m\}$ and $\alpha'<\alpha$ it holds $F_{\mathrm{Beta}(s-\gamma+1,1)}^{-1}(1-\alpha')\geq F_{\mathrm{Beta}(s-\gamma+1,1)}^{-1}(1-\alpha)$ and therefore $\Gamma_{j}(\alpha')\subseteq\Gamma_{j}(\alpha)$.

To see that $(GA4)$ is fulfilled, let $j\in\{1,\ldots,m\}$ be fixed. We calculate the set of LFCs for $\varphi_{j}$, i.e. the set of parameters $\vartheta\in\Theta$ that yield the supremum in
\begin{equation*}
\underset{\vartheta'\in\Theta:\theta_{j}(\vartheta')\in H_{j}}{\mathrm{sup}}\text{pr}_{\vartheta'}\big(T_{j}(X)\in\Gamma_{j}(\alpha)\big),
\end{equation*}
and show that it does not depend on $\alpha$.

First, it holds $\text{pr}_{\vartheta}\big(T_{j}(X)\in\Gamma_{j}(\alpha)\big)=\text{pr}_{\vartheta}\big(1-p_{(\gamma),j}(X)>F_{\mathrm{Beta}(s-\gamma+1,1)}^{-1}(1-\alpha)\big)$, which is larger the smaller the $p$-values $p_{1,j}(X),\ldots,p_{s,j}(X)$ (stochastically) are. For every $i=1,\ldots,s$, due to $(RA4)$, there exist parameters $\vartheta_{i}\in\Theta$, independent of $\alpha$, such that $p_{i,j}(X)=0$ almost surely under $\vartheta_{j}$. Independently of $\alpha$, this is satisfied for parameters with $\theta_{i,j}(\vartheta_{i})$ large enough. It is clear, that for any LFC $\vartheta_{0}\in\Theta$ for $\varphi_{j}$, it has to hold $\theta_{j}(\vartheta_{0})\in H_{j}$ and $\theta_{i,j}(\vartheta_{0})$ large enough (without loss of generality equal to $\infty$) for $\gamma-1$ indices $i$.

Without loss of generality, we consider a parameter $\vartheta_{0}\in\Theta$ with $\theta_{j}(\vartheta_{0})\in H_{j}$ and $\theta_{1,j}(\vartheta_{0})=\cdots=\theta_{\gamma-1,j}(\vartheta_{0})=\infty$, leaving $\theta_{i,j}(\vartheta_{0})\leq 0$ for the remaining indices $i=\gamma,\ldots,s$.

Due to assumption $(RA4)$, the $p$-values $p_{1,j}(X),\ldots,p_{\gamma-1,j}(X)$ are equal to zero and $T_{j}(X)=1-p_{(\gamma),j}(X)=\mathrm{max}\{1-p_{\gamma,j}(X),\ldots,1-p_{s,j}(X)\}$ almost surely under $\vartheta_{0}$. We obtain that
\begin{align}
&\text{pr}_{\vartheta_{0}}\big(T_{j}(X)\in\Gamma_{j}(\alpha)\big)\nonumber\\
=&\text{pr}_{\vartheta_{0}}\big(\mathrm{max}\{1-p_{\gamma,j}(X),\ldots,1-p_{s,j}(X)\}>F_{\mathrm{Beta}(s-\gamma+1,1)}^{-1}(1-\alpha)\big)\nonumber\\
=&1-\text{pr}_{\vartheta_{0}}\big(\mathrm{max}\{1-p_{\ell,j}(X): \gamma \leq \ell \leq s\} 
\leq F_{\mathrm{Beta}(s-\gamma+1,1)}^{-1}(1-\alpha)\big)\label{eq:replicabilityga1}.
\end{align}
Since the studies are independent, \eqref{eq:replicabilityga1} is equal to
\begin{align}
&1-\prod_{i=\gamma}^{s}\text{pr}_{\vartheta_{0}}\big(1-p_{i,j}(X)\leq
F_{\mathrm{Beta}(s-\gamma+1,1)}^{-1}(1-\alpha)\big)\nonumber\\
=&1-\prod_{i=\gamma}^{s}\Big[1-\text{pr}_{\vartheta_{0}}\big(p_{i,j}(X)<1-F_{\mathrm{Beta}(s-\gamma+1,1)}^{-1}(1-\alpha)\big)\Big]\label{eq:replicabilityga2}.
\end{align}
Now, using the relation in \eqref{eq:newLFCGamma}, the term in \eqref{eq:replicabilityga2} equals
\begin{equation*}
1-\prod_{i=u}^{s}\Big[1-\text{pr}_{\vartheta_{0}}\big(T_{i,j}(X)\in\Gamma_{i,j}(\alpha_{i,j})\big)\Big],
\end{equation*}
where $\alpha_{i,j}=1-F_{\mathrm{Beta}(s-u+1,1)}^{-1}(1-\alpha)$, which is maximized if each term $\text{pr}_{\vartheta_{0}}\big(T_{i,j}(X)\in\Gamma_{i,j}(\alpha_{i,j})\big)$ is maximized over the set of all $\vartheta_{0}:\theta_{i,j}(\vartheta_{0})\leq 0\ (i=\gamma,\ldots,s)$. Due to assumption $(RA2)$, this is the case for any $\vartheta_{0}\in\Theta$ with $\theta_{i,j}(\vartheta_{0})=0$ independently of $\alpha_{i,j}\ (i=\gamma,\ldots,s)$, such that $\text{pr}_{\vartheta_{0}}\big(T_{j}(X)\in\Gamma_{j}(\alpha)\big)$ is being maximized by any parameter $\vartheta_{0}$ with
\begin{equation*}
\theta_{j}(\vartheta_{0})=(\underbrace{\infty,\ldots,\infty}_{\gamma-1},\underbrace{0,\ldots,0}_{s-\gamma+1})
\end{equation*}  
independently of $\alpha$.

Altogether, the set of LFCs for $\varphi_{j}$ is
\begin{equation*}
\{\vartheta\in\Theta:\theta_{j}(\vartheta)\;\text{is any permutation of}\;(\underbrace{\infty,\ldots,\infty}_{\gamma-1},\underbrace{0,\ldots,0}_{s-\gamma+1})\},
\end{equation*}
hence, obviously independent of $\alpha$.

Finally, we verify $(GA3)$ as follows: For every $j\in \{1,\ldots,m\}$ and $\alpha\in(0,1)$, it holds 
\begin{align} 
&\;\underset{\vartheta\in\Theta:\theta_{j}(\vartheta)\in H_{j}}{\mathrm{sup}}\text{pr}_{\vartheta}(T_{j}(X)\in\Gamma_{j}(\alpha))=\text{pr}_{\vartheta_{0}}\big(T_{j}(X)\in\Gamma_{j}(\alpha)\big)\\
=&\text{pr}_{\vartheta_{0}}\big[\mathrm{max}\{1-p_{\ell,j}(X): \gamma \leq \ell \leq s\}\geq F_{\mathrm{Beta}(s-\gamma+1,1)}^{-1}(1-\alpha)\big],\label{eq:replicabilityga3}
\end{align}
due to $(RA4)$, where $\vartheta_{0}\in\Theta$ with $\theta_{j}(\vartheta_{0})=(\infty,\ldots,\infty,0,\ldots,0)$ is an LFC for $\varphi_{j}$.

Furthermore, $1-p_{i,j}(X)$ is uniformly distributed on $[0,1]$ under an LFC $\vartheta_{0}\in\Theta$ with $\theta_{i,j}(\vartheta_{0})=0\ (i=\gamma,\ldots,s)$. Since $\mathrm{max}(U_{1},\ldots,U_{k})$ is  
$\mathrm{Beta}(s-\gamma+1,1)$-distributed, for $U_{1},\ldots,U_{k}$, that are stochastically independent and identically, uniformly distributed on $[0,1]$, we obtain that 
\eqref{eq:replicabilityga3} equals $1-F_{\mathrm{Beta}(s-\gamma+1,1)}\big(F_{\mathrm{Beta}(s-\gamma+1,1)}^{-1}(1-\alpha)\big)=\alpha$, as desired.

\subsection*{Proof of Theorem~\ref{thm:4}}
\label{app:proofthm4}
We want to show, that
\begin{equation}
\label{eq:thm4eq1}
T_{j}(X)^{(\vartheta)}\leq_{\mathrm{hr}}T_{j}(X)^{(\vartheta_{0})}
\end{equation}  
holds for any parameters $\vartheta\in\Theta$ with $\theta_{j}(\vartheta)\in H_{j}$ and $\vartheta_{0}$ an LFC for $\varphi_{j}$. Let $\vartheta\in\Theta$ with $\theta_{j}(\vartheta)\in H_{j}$, i.e. $\theta_{i,j}(\vartheta)\leq 0$ for at least $s-\gamma+1$ indices $i$, be given. Since the distribution of $T_{j}(X)$ does not depend on the particular form of the LFC $\vartheta_{0}$, we choose an LFC that fulfills $\theta_{i,j}(\vartheta)\leq 0=\theta_{i,j}(\vartheta_{0})$ for at least $s-\gamma+1$ indices $i$. Without loss of generality, let $\theta_{i,j}(\vartheta)\leq 0\ (i=1,\ldots,s-\gamma+1)$, and 
\begin{equation*}
\theta_{j}(\vartheta_{0})=(\underbrace{0,\ldots,0}_{s-\gamma+1},\underbrace{\infty,\ldots,\infty}_{\gamma-1}).
\end{equation*}

For $i=1,\ldots,s-\gamma+1$, it is $\theta_{i,j}(\vartheta)\leq 0=\theta_{i,j}(\vartheta_{0})$, and therefore $T_{i,j}(X)^{(\vartheta)}\leq_{\mathrm{hr}}T_{i,j}(X)^{(\vartheta_{0})}$. Let $F_{i,j}$ be the cumulative distribution function of $T_{i,j}(X)$ under an LFC for $\varphi_{i,j}$, i.e. under a $\tilde{\vartheta}\in\Theta$ with $\theta_{i,j}(\tilde{\vartheta})=0$. For $i=1,\ldots,s-\gamma+1$, it holds $\theta_{i,j}(\vartheta_{0})=0$, i.e. the parameter $\vartheta_{0}$ is an LFC for $\varphi_{i,j}\ (i=1,\ldots.,s-\gamma+1)$. From Part $3$ in Theorem~\ref{thm:hazardorder}, it follows that
\begin{equation}
\label{eq:thm4eq2}
1-p_{i,j}(X)=F_{i,j}\big(T_{i,j}(X)\big)^{(\vartheta)}\leq_{\mathrm{hr}}F_{i,j}\big(T_{i,j}(X)\big)^{(\vartheta_{0})}.
\end{equation}
Note that $F_{i,j}\big(T_{i,j}(X)\big)$ is uniformly distributed on $[0,1]$ under $\vartheta_{0}$, $(i=1,\ldots,s-u+1)$.

For ease of notation, we write $P_{i}=1-p_{i,j}$ and $T_{j}(X)=1-p_{(\gamma),j}(X)=P_{(s-\gamma+1)}(X)$. Under $\vartheta_{0}$ it then holds $T_{j}(X)$ and $\mathrm{max}\{U_{1},\ldots,U_{s-\gamma+1}\}$ are identically distributed, since $P_{s-\gamma+2}(X)=\cdots=P_{s}(X)=1$ almost surely due to $(RA4)$, where $U_{1},\ldots,U_{s-\gamma+1}$ are stochastically independent and identically, uniformly distributed on $[0,1]$.

Now, \eqref{eq:thm4eq1} is equivalent to $P_{(s-\gamma+1:n)}(X)^{(\vartheta)}\leq_{\mathrm{hr}}P_{(s-\gamma+1:s)}(X)^{(\vartheta_{0})}\sim U_{(s-\gamma+1:s-\gamma+1)}$, which follows directly from Lemma~\ref{lm:hazard}, since, from \eqref{eq:thm4eq2}, it holds $P_{i}(X)^{(\vartheta)}\leq_{\mathrm{hr}}U_{i}$ for at least $s-\gamma+1$ indices $i\in\{1,\ldots,s\}$.

\section{Further simulation results}
\label{sup:sec:furthersimulations}

The results of our Monte Carlo simulation with regard to the standard deviations, cf. the end of 
Section~\ref{sec:pi0simulation}, are listed in Table~\ref{tab:pi0stdLFC} and Table~\ref{tab:pi0stdrand} for the utilization of the LFC-based and the randomized $p$-values, respectively.

Furthermore, we looked at two different approaches for defining the LFC-based $p$-values. The test statistics $T_{j}(X)=1-p_{(\gamma),j}$ do not regard the size of the $s-\gamma$ larger $p$-values $p_{(\gamma+1)},\ldots,p_{(s)}$ explicitly. Instead, one could consider
\begin{equation*}
T_{j}^{(S)}(X)=(s-\gamma+1)^{-1/2}\sum_{i=\gamma}^{s}\Phi^{-1}\big(1-p_{(i),j}(X)\big),\quad\Gamma_{j}^{(S)}(\alpha)=\big(\Phi^{-1}(1-\alpha),\infty\big),
\end{equation*}
or 
\begin{equation*}
T_{j}^{(F)}(X)=-2\sum_{i=\gamma}^{s}\mathrm{log}(p_{(i),j}(X)),\quad\Gamma_{j}^{(F)}(\alpha)=\big(F_{\chi^{2}_{2~(s-\gamma+1)}}^{-1}(1-\alpha),\infty\big),
\end{equation*}
motivated by the Stouffer method and the Fisher method for combining $p$-values, respectively, where $\Phi$ is the cumulative distribution function of the standard normal distribution in $\mathbb{R}$, and $F_{\chi^{2}_{2~(s-\gamma+1)}}$ is the cumulative distribution function of a $\chi^{2}$-distribution with $2~(s-\gamma+1)$ degrees of freedom \citep[Sec. 2.2]{benjamini2008screening}. \cite{benjamini2008screening} showed that applying the Benjamini--Hochberg linear step up test from \cite{benjamini1995controlling} on the LFC-based $p$-values $p_{1}^{LFC},\ldots,p_{m}^{LFC}$ controls the false discovery rate even if the $p$-values within each study admit a positive dependence. For more details see Theorem $3$ in \cite{benjamini2008screening}.

Models based on these test statistics, however, do not fulfill assumption $(GA1)$ from Section~\ref{sec:modelsetup}, such that Theorem~\ref{thm:1} does not apply, and calculating the randomized $p$-values $p_{1}^{rand},\ldots,p_{m}^{rand}$ as in Definition~\ref{def:randomized} becomes more difficult.

We simulated the expected values of the estimator $\hat{\pi}_{0}(1/2)$ when utilizing the LFC-based $p$-values under these alternative test statistics. The results of the Monte Carlo simulations with $10,000$ repetitions can be found in Table~\ref{tab:stouffer} for the Stouffer-based and Table~\ref{tab:fisher} for the Fisher-based $p$-values. More accurate estimations as compared to $\hat{\pi}_{0}^{rand}$ are written in bold. Compared to the expected values when utilizing our randomized $p$-values $(p_{j}^{rand})_{j}$ both alternatives only perform better in case of $\mu_{\mathrm{min}}=0$ and lower $\gamma$ ($2,4,6$ for Stouffer, and $2,4$ for Fisher).

\begin{table}[htbp]
\caption{Empirical standard deviations for $\hat{\pi}_{0}(1/2)$ using $(p_{j}^{LFC})_{j=1,\ldots,m}$ in Model 1 with $s=10$, resulting from a Monte Carlo simulation with $10,000$ repetitions}
  \centering
    \begin{tabular}{lrrrrr}
    $\gamma=2$   & \multicolumn{1}{l}{$\pi_{0}$} & 0.6   & 0.7   & 0.8   & 0.9 \\
    $(\mu_{\mathrm{min}},\mu_{\mathrm{max}})$ &       &       &       &       &  \\
    (0,2) &       & 0.07582752 & 0.08296189 & 0.08786311 & 0.09318406 \\
    (-0.5,3) &       & 0.06955630 & 0.07470407 & 0.08089620 & 0.08561848 \\
    (-1,4) &       & 0.06144261 & 0.06616704 & 0.07076652 & 0.07451185 \\
    (-1.5,5) &       & 0.05453308 & 0.05847278 & 0.06289423 & 0.06657209 \\
          &       &       &       &       &  \\
    $\gamma=4$   & \multicolumn{1}{l}{$\pi_{0}$} & 0.6   & 0.7   & 0.8   & 0.9 \\
    $(\mu_{\mathrm{min}},\mu_{\mathrm{max}})$ &       &       &       &       &  \\
    (0,2) &       & 0.06851676 & 0.07402335 & 0.07861075 & 0.08404400 \\
    (-0.5,3) &       & 0.05976733 & 0.06390461 & 0.06858124 & 0.07387649 \\
    (-1,4) &       & 0.05125430 & 0.05480785 & 0.05943488 & 0.06301813 \\
    (-1.5,5) &       & 0.04496894 & 0.04844289 & 0.05183053 & 0.05536219 \\
          &       &       &       &       &  \\
    $\gamma=6$   & \multicolumn{1}{l}{$\pi_{0}$} & 0.6   & 0.7   & 0.8   & 0.9 \\
    $(\mu_{\mathrm{min}},\mu_{\mathrm{max}})$ &       &       &       &       &  \\
    (0,2) &       & 0.06060406 & 0.06547952 & 0.06891937 & 0.07425530 \\
    (-0.5,3) &       & 0.05168182 & 0.05610461 & 0.05966287 & 0.06261066 \\
    (-1,4) &       & 0.04385788 & 0.04703952 & 0.05022082 & 0.05353003 \\
    (-1.5,5) &       & 0.03789674 & 0.04156683 & 0.04366071 & 0.04636814 \\
          &       &       &       &       &  \\
    $\gamma=8$   & \multicolumn{1}{l}{$\pi_{0}$} & 0.6   & 0.7   & 0.8   & 0.9 \\
    $(\mu_{\mathrm{min}},\mu_{\mathrm{max}})$ &       &       &       &       &  \\
    (0,2) &       & 0.06101877 & 0.06274993 & 0.06467866 & 0.06596085 \\
    (-0.5,3) &       & 0.05000569 & 0.05270983 & 0.05467924 & 0.05597926 \\
    (-1,4) &       & 0.04252793 & 0.04467561 & 0.04602834 & 0.04805289 \\
    (-1.5,5) &       & 0.03655808 & 0.0375914 & 0.03910888 & 0.04120146 \\
          &       &       &       &       &  \\
    $\gamma=10$  & \multicolumn{1}{l}{$\pi_{0}$} & 0.6   & 0.7   & 0.8   & 0.9 \\
    $(\mu_{\mathrm{min}},\mu_{\mathrm{max}})$ &       &       &       &       &  \\
    (0,2) &       & 0.07542883 & 0.07593555 & 0.07524137 & 0.07511923 \\
    (-0.5,3) &       & 0.06533901 & 0.06485274 & 0.06546255 & 0.06497500 \\
    (-1,4) &       & 0.05600993 & 0.05558931 & 0.05601106 & 0.05565804 \\
    (-1.5,5) &       & 0.04845610 & 0.04821363 & 0.04870596 & 0.04752345 \\
    \end{tabular}%
  \label{tab:pi0stdLFC}%
\end{table}%
\begin{table}[htbp]
\caption{Empirical standard deviations for $\hat{\pi}_{0}(1/2)$ using $(p_{j}^{rand})_{j=1,\ldots,m}$ in Model $1$ with $s=10$, resulting from a Monte Carlo simulation with $10,000$ repetitions}
  \centering
    \begin{tabular}{lrrrrr}
    $\gamma=2$   & \multicolumn{1}{l}{$\pi_{0}$} & 0.6   & 0.7   & 0.8   & 0.9 \\
    $(\mu_{\mathrm{min}},\mu_{\mathrm{max}})$ &       &       &       &       &  \\
    (0,2) &       & 0.07581577 & 0.08296278 & 0.08805736 & 0.09328741 \\
    (-0.5,3) &       & 0.0702544 & 0.07536504 & 0.08154187 & 0.08638383 \\
    (-1,4) &       & 0.06395161 & 0.06911728 & 0.07406993 & 0.07819693 \\
    (-1.5,5) &       & 0.06221166 & 0.06620503 & 0.07202633 & 0.07624616 \\
          &       &       &       &       &  \\
    $\gamma=4$   & \multicolumn{1}{l}{$\pi_{0}$} & 0.6   & 0.7   & 0.8   & 0.9 \\
    $(\mu_{\mathrm{min}},\mu_{\mathrm{max}})$ &       &       &       &       &  \\
    (0,2) &       & 0.06958866 & 0.07479425 & 0.07966822 & 0.08502931 \\
    (-0.5,3) &       & 0.06402867 & 0.06926745 & 0.07450052 & 0.07952569 \\
    (-1,4) &       & 0.0627167 & 0.06769249 & 0.07350674 & 0.07756825 \\
    (-1.5,5) &       & 0.06598907 & 0.07154958 & 0.07621329 & 0.08084872 \\
          &       &       &       &       &  \\
    $\gamma=6$   & \multicolumn{1}{l}{$\pi_{0}$} & 0.6   & 0.7   & 0.8   & 0.9 \\
    $(\mu_{\mathrm{min}},\mu_{\mathrm{max}})$ &       &       &       &       &  \\
    (0,2) &       & 0.06774445 & 0.07290022 & 0.07769824 & 0.0821969 \\
    (-0.5,3) &       & 0.06669294 & 0.07219919 & 0.077621 & 0.08185157 \\
    (-1,4) &       & 0.07012927 & 0.07496831 & 0.08007004 & 0.08618822 \\
    (-1.5,5) &       & 0.07328739 & 0.07785646 & 0.08384216 & 0.08907628 \\
          &       &       &       &       &  \\
    $\gamma=8$   & \multicolumn{1}{l}{$\pi_{0}$} & 0.6   & 0.7   & 0.8   & 0.9 \\
    $(\mu_{\mathrm{min}},\mu_{\mathrm{max}})$ &       &       &       &       &  \\
    (0,2) &       & 0.07937488 & 0.08193387 & 0.0860587 & 0.0897265 \\
    (-0.5,3) &       & 0.0774249 & 0.08324976 & 0.08713005 & 0.09173463 \\
    (-1,4) &       & 0.07788223 & 0.08412603 & 0.08823772 & 0.09349471 \\
    (-1.5,5) &       & 0.07831725 & 0.0835687 & 0.08916617 & 0.09328771 \\
          &       &       &       &       &  \\
    $\gamma=10$  & \multicolumn{1}{l}{$\pi_{0}$} & 0.6   & 0.7   & 0.8   & 0.9 \\
    $(\mu_{\mathrm{min}},\mu_{\mathrm{max}})$ &       &       &       &       &  \\
    (0,2) &       & 0.09869885 & 0.09828208 & 0.09745444 & 0.09849114 \\
    (-0.5,3) &       & 0.09665863 & 0.09552648 & 0.0971173 & 0.09756151 \\
    (-1,4) &       & 0.09492729 & 0.09507453 & 0.09681341 & 0.09597547 \\
    (-1.5,5) &       & 0.09626054 & 0.09466757 & 0.09469119 & 0.09489155 \\
    \end{tabular}%
  \label{tab:pi0stdrand}%
\end{table}%
\begin{table}[htbp]
    \caption{Expected values of $\hat{\pi}_{0}(1/2)$ using $(p_{j}^{(S)})_{j=1,\ldots,m}$ under Model $1$ with $s=10$. Values result from Monte Carlo simulations with $10,000$ repetitions. Values that come closer to the true proportion $\pi_{0}$ than under the use of our randomized $p$-values are written in bold.}
  \centering
    \begin{tabular}{lrrrrr}
    $\gamma=2$   & \multicolumn{1}{l}{$\pi_{0}$} & 0.6   & 0.7   & 0.8   & 0.9 \\
    $(\mu_{\mathrm{min}},\mu_{\mathrm{max}})$ &       &       &       &       &  \\
    (-0,2) &       & \textbf{0.6646} & \textbf{0.7746} & \textbf{0.8859} & \textbf{0.9968} \\
    (-0.5,3) &       & 0.9636 & 1.1246 & 1.2855 & 1.4458 \\
    (-1,4) &       & 1.1254 & 1.3129 & 1.5002 & 1.6878 \\
    (-1.5,5) &       & 1.1809 & 1.3777 & 1.5746 & 1.7714 \\
          &       &       &       &       &  \\
    $\gamma=4$   & \multicolumn{1}{l}{$\pi_{0}$} & 0.6   & 0.7   & 0.8   & 0.9 \\
    $(\mu_{\mathrm{min}},\mu_{\mathrm{max}})$ &       &       &       &       &  \\
    (-0,2) &       & \textbf{0.7806} & \textbf{0.9095} & \textbf{1.0402} & \textbf{1.1699} \\
    (-0.5,3) &       & 1.005 & 1.1721 & 1.3395 & 1.5072 \\
    (-1,4) &       & 1.1287 & 1.3166 & 1.5047 & 1.6928 \\
    (-1.5,5) &       & 1.1775 & 1.3737 & 1.5697 & 1.7661 \\
          &       &       &       &       &  \\
    $\gamma=6$   & \multicolumn{1}{l}{$\pi_{0}$} & 0.6   & 0.7   & 0.8   & 0.9 \\
    $(\mu_{\mathrm{min}},\mu_{\mathrm{max}})$ &       &       &       &       &  \\
    (-0,2) &       & \textbf{0.8837} & \textbf{1.0281} & \textbf{1.1737} & \textbf{1.3187} \\
    (-0.5,3) &       & 1.0401 & 1.2114 & 1.3836 & 1.5556 \\
    (-1,4) &       & 1.1322 & 1.3196 & 1.5077 & 1.6947 \\
    (-1.5,5) &       & 1.1742 & 1.3692 & 1.5636 & 1.7582 \\
          &       &       &       &       &  \\
    $\gamma=8$   & \multicolumn{1}{l}{$\pi_{0}$} & 0.6   & 0.7   & 0.8   & 0.9 \\
    $(\mu_{\mathrm{min}},\mu_{\mathrm{max}})$ &       &       &       &       &  \\
    (-0,2) &       & 1.0014 & 1.1511 & 1.301 & 1.4509 \\
    (-0.5,3) &       & 1.0872 & 1.2572 & 1.4286 & 1.5989 \\
    (-1,4) &       & 1.1484 & 1.3305 & 1.5131 & 1.6956 \\
    (-1.5,5) &       & 1.1815 & 1.3704 & 1.5593 & 1.7484 \\
          &       &       &       &       &  \\
    $\gamma=10$  & \multicolumn{1}{l}{$\pi_{0}$} & 0.6   & 0.7   & 0.8   & 0.9 \\
    $(\mu_{\mathrm{min}},\mu_{\mathrm{max}})$ &       &       &       &       &  \\
    (-0,2) &       & 1.2141 & 1.3332 & 1.4525 & 1.5724 \\
    (-0.5,3) &       & 1.2064 & 1.3501 & 1.4938 & 1.6377 \\
    (-1,4) &       & 1.2167 & 1.3752 & 1.5335 & 1.6919 \\
    (-1.5,5) &       & 1.2243 & 1.3924 & 1.5608 & 1.7285 \\
    \end{tabular}%
  \label{tab:stouffer}%
\end{table}%
\begin{table}[htbp]
    \caption{Expected values of $\hat{\pi}_{0}(1/2)$ using $(p_{j}^{(F)})_{j=1,\ldots,m}$ under Model $1$ with $s=10$. Values result from Monte Carlo simulations with $10,000$ repetitions. Values that come closer to the true proportion $\pi_{0}$ than under the use of our randomized $p$-values are in bold}
  \centering
    \begin{tabular}{lrrrrr}
    $\gamma=2$   & \multicolumn{1}{l}{$\pi_{0}$} & 0.6   & 0.7   & 0.8   & 0.9 \\
    $(\mu_{\mathrm{min}},\mu_{\mathrm{max}})$ &       &       &       &       &  \\
    (-0,2) &       & \textbf{0.6895} & \textbf{0.8036} & \textbf{0.9193} & \textbf{1.0347} \\
    (-0.5,3) &       & 0.9489 & 1.1068 & 1.2652 & 1.4232 \\
    (-1,4) &       & 1.0946 & 1.2773 & 1.4599 & 1.6418 \\
    (-1.5,5) &       & 1.1567 & 1.3497 & 1.5426 & 1.7353 \\
          &       &       &       &       &  \\
    $\gamma=4$   & \multicolumn{1}{l}{$\pi_{0}$} & 0.6   & 0.7   & 0.8   & 0.9 \\
    $(\mu_{\mathrm{min}},\mu_{\mathrm{max}})$ &       &       &       &       &  \\
    (-0,2) &       & \textbf{0.8337} & \textbf{0.9716} & \textbf{1.1106} & \textbf{1.2498} \\
    (-0.5,3) &       & 1.0134 & 1.182 & 1.3509 & 1.5201 \\
    (-1,4) &       & 1.1143 & 1.2997 & 1.4859 & 1.6715 \\
    (-1.5,5) &       & 1.1606 & 1.3536 & 1.547 & 1.7408 \\
          &       &       &       &       &  \\
    $\gamma=6$   & \multicolumn{1}{l}{$\pi_{0}$} & 0.6   & 0.7   & 0.8   & 0.9 \\
    $(\mu_{\mathrm{min}},\mu_{\mathrm{max}})$ &       &       &       &       &  \\
    (-0,2) &       & 0.9427 & 1.0971 & 1.2528 & 1.4084 \\
    (-0.5,3) &       & 1.0593 & 1.2345 & 1.4104 & 1.5858 \\
    (-1,4) &       & 1.1282 & 1.3158 & 1.5037 & 1.6907 \\
    (-1.5,5) &       & 1.1626 & 1.3561 & 1.5492 & 1.7426 \\
          &       &       &       &       &  \\
    $\gamma=8$   & \multicolumn{1}{l}{$\pi_{0}$} & 0.6   & 0.7   & 0.8   & 0.9 \\
    $(\mu_{\mathrm{min}},\mu_{\mathrm{max}})$ &       &       &       &       &  \\
    (-0,2) &       & 1.0495 & 1.2077 & 1.3659 & 1.5243 \\
    (-0.5,3) &       & 1.1063 & 1.2811 & 1.4565 & 1.6315 \\
    (-1,4) &       & 1.1486 & 1.3328 & 1.5175 & 1.7021 \\
    (-1.5,5) &       & 1.1716 & 1.3616 & 1.5514 & 1.7416 \\
          &       &       &       &       &  \\
    $\gamma=10$  & \multicolumn{1}{l}{$\pi_{0}$} & 0.6   & 0.7   & 0.8   & 0.9 \\
    $(\mu_{\mathrm{min}},\mu_{\mathrm{max}})$ &       &       &       &       &  \\
    (-0,2) &       & 1.2141 & 1.3332 & 1.4525 & 1.5724 \\
    (-0.5,3) &       & 1.2064 & 1.3501 & 1.4938 & 1.6377 \\
    (-1,4) &       & 1.2167 & 1.3752 & 1.5335 & 1.6919 \\
    (-1.5,5) &       & 1.2243 & 1.3924 & 1.5608 & 1.7285 \\
    \end{tabular}%
  \label{tab:fisher}%
\end{table}%

\newpage
\bibliographystyle{elsarticle-harv}
\bibliography{Replicability-arXiv}

\end{document}